\documentclass[12pt]{article}
\usepackage{amsmath}
\usepackage{graphicx}
\usepackage{enumerate}
\usepackage{natbib}
\setcitestyle{authoryear,aysep={}}
\usepackage{url} 
\usepackage{amssymb}
\usepackage{amsthm}
\newtheorem{theorem}{Theorem}

\newcommand{\blind}{1}

\begin{document}

\def\spacingset#1{\renewcommand{\baselinestretch}%
{#1}\small\normalsize} \spacingset{1.5}


\if1\blind
{
  \title{\bf Modeling Preferences: A Bayesian Mixture of Finite Mixtures for Rankings and Ratings}
  \author{Michael Pearce$^\text{a}$ and Elena A. Erosheva$^\text{abc}$\\
  \\
  $^\text{a}$Department of Statistics, University of Washington\\
  $^\text{b}$School of Social Work, University of Washington\\
  $^\text{c}$Center for Statistics and the Social Sciences, University of Washington}
  \maketitle
} \fi

\if0\blind
{
  \bigskip
  \bigskip
  \bigskip
  \begin{center}
    {\LARGE\bf Modeling Preferences: A Bayesian Mixture of Finite Mixtures Approach to Rankings and Ratings}
\end{center}
  \medskip
} \fi

\bigskip
\begin{abstract}
Rankings and ratings are commonly used to express preferences but provide distinct and complementary information. Rankings give ordinal and scale-free comparisons but lack granularity; ratings provide cardinal and granular assessments but may be highly subjective or inconsistent. Collecting and analyzing rankings and ratings jointly has not been performed until recently due to a lack of principled methods. In this work, we propose a flexible, joint statistical model for rankings and ratings under heterogeneous preferences: the Bradley-Terry-Luce-Binomial (BTL-Binomial). We employ a Bayesian mixture of finite mixtures (MFM) approach to estimate heterogeneous preferences, understand their inherent uncertainty, and make accurate decisions based on ranking and ratings jointly. We demonstrate the efficiency and practicality of the BTL-Binomial MFM approach on real and simulated datasets of ranking and rating preferences in peer review and survey data contexts.
\end{abstract}

\noindent%
\vfill

\newpage
\spacingset{1.5} 

\section{Introduction}\label{sect:intro}

Rankings and ratings allow for expressing preferences on a collection of objects. Rankings of objects are relative judgements that utilize other objects as reference points. These judgements may come in many forms, such as ordered lists of complete or partial rankings, as well as pairwise or groupwise comparisons. While rankings provide ordinal and scale-free comparisons between objects, the information they provide is local in the sense that it is impossible to tell, for example, whether someone's preferences for first and second place are about equal or drastically different. On the other hand, ratings (often referred to as \textit{scores}) are numerical values that denote an object's perceived quality. Ratings are absolute judgements in the sense that ratings do not directly use other objects as reference points. Instead, ratings reflect preferences in relation to some standard or target level of performance that is indicated with verbal descriptions of a scale: low to high, poor to excellent, etc.  Ratings can  provide cardinal and granular assessments between objects and, when calibrated, allow for global comparisons. However, common ties when rating two or more objects limit the use of ratings for demarcating quality \citep{shah2018design}. Ratings may also be highly subjective or inconsistent \citep{biernat1995shifting,Wang2018}. As such, rankings and ratings provide distinct and complementary information on preferences.

Psychological and psychometric literatures have long been documenting properties and deficiencies of rankings and ratings as expressions of preferences. Rankings have been criticized for imposing high cognitive load, especially when the number of options is large \citep{alwin1985measurement}; for potentially forcing judges to make invalid distinctions in cases with low discriminability \citep{russell1994ranking}; and for being difficult to analyze by means of common statistical techniques \citep{sung2018visual}.  Ratings have been criticized for providing measurements that are coarse and for allowing judges to use the same numeric value for more than one case \citep{russell1994ranking}. Comparative judgements are less susceptible to noise: research on job performance measurement, measurement of attitudes, and person perception
has found that rankings have better validity and may have more accuracy than absolute ratings \citep{goffin2011all}.
To reconcile and draw on both approaches, psychological research has moved to suggest variations that try to combine good qualities of both expressions at the expense of more elaborate data collection mechanisms compared to either rankings or ratings (e.g., \cite{sung2018visual}).

In recent years, growing literatures in computer science and social science disciplines have suggested that collecting and analyzing both rankings and ratings may be beneficial to understanding preferences. That is, using rankings and ratings in tandem may retain the benefits and minimize the downsides of each data type \citep{ovadia2004ratings,
shah2018design,liu2022integrating}.
However, few such methods exist. Motivated by  problems in meta-search, information retrieval, and peer review, authors in computer science have proposed algorithmic approaches for merging ordinal and cardinal preferences \citep{ailon2010aggregation, 
liu2022integrating}. These works do not allow for the quantification of uncertainty. Furthermore, many assume that a judge's rankings and ratings must be internally consistent, which is often unrealistic and has to be enforced during data collection. Alternatively, authors in health sciences \citep{Salomon2003,Kim2015} 
rely on converting ordinal and cardinal data of one type into another (e.g., converting scores into a ranking via simple ordering). Conversion is suboptimal because it may distort or discard the observed data.

To the best of our knowledge, Mallows-Binomial \citep{pearce2022unified} is the first and only joint statistical model for rankings and ratings that does not rely on data conversion. The model shares parameters between Mallows and Binomial distributions for rankings and ratings, respectively, which may be used to perform inference on the absolute and relative qualities of objects. This model does not require judges to provide internally consistent rankings and ratings. However, it has three major drawbacks that limit practical applicability. First, it requires rankings be either top-$r$ or complete lists of objects. Thus, Mallows-Binomial cannot accommodate pairwise comparisons or when different judges have access to different sets of objects (the ``separate ballots" problem). Second, Mallows-Binomial does not allow for heterogeneity. Third, rankings that follow a Mallows distribution do not satisfy Luce's Choice Axiom  \citep{luce1959individual}.


The work in this paper is motivated by the following three applied settings. In each, rankings and ratings are of practical use, yet cannot be analyzed by the Mallows-Binomial. 

\vspace{-0.2in}
\paragraph{Setting 1: Paper Selection in Large Academic Conferences} We first consider quality assessment of papers to large academic conferences where no single reviewer evaluates all papers. This situation closely mirrors that studied by \cite{liu2022integrating} in computer science. We suppose a large number of papers are submitted to a conference and only a small subset may be accepted. To make decisions, each paper is assigned a few reviewers and each reviewer is assigned a few papers to review, such that reviewers generally assess overlapping subsets of papers. A simple method for collecting preferences is to ask each reviewer to rate each paper on a clearly-defined discrete  ``common scale''. Then, the papers with the best average ratings are selected for acceptance. However, this approach is suboptimal because (1) ratings may be inconsistent since reviewers may interpret the scale in unique ways \citep{baumgartner2001response}, and (2) delineating the papers based on average ratings may be impossible or imprecise since average ratings can produce ties or near ties, especially when the scale is coarse or the paper is assessed few times. We note that finding additional reviewers may be impractical. In addition, increasing the granularity of the rating scale may not increase precision, but instead increase noise \citep{miller1956magical}.
In this paper, we demonstrate that these problems can be addressed by the introduction of rankings. Both NeurIPS 2016 and ICML 2021 collected rankings from reviewers in addition to ratings \citep{shah2018design}. Still, a principled statistical method by which to incorporate rankings and ratings jointly that could be applicable to such settings does not exist. Because such conferences typically handle high volumes of paper submissions \citep{shah2018design}, we focus only on point estimation of paper quality and do not estimate any potential heterogeneity.

\vspace{-0.2in}
\paragraph{Setting 2: Proposal Selection in Grant Panel Review under Heterogeneity} In grant panel review, reviewers evaluate a small number of grant proposals with respect to some criteria such as the scientific merit. Often, mean ratings are used to communicate proposal quality for funding decisions, even though they exhibit similar problems to those described in the previous setting. At a time of funding scarcity--for example, R01 research award rates at the National Institutes of Health vary between 10 and 20\% \citep{erosheva2020nih}--it is most important to obtain clear and accurate demarcation of proposals at the top. As such, the addition of top-$r$ rankings, in which $r$ is slightly larger than the number of proposals to be funded, may be useful. Top-$r$ rankings provide additional information on the ``best" proposals without creating a substantial cognitive burden on reviewers to rank each and every proposal. \cite{pearce2022unified} studied this setting using a frequentist joint ranking-rating Mallows-Binomial model, under the assumption of a single ground-truth ranking of proposals (i.e., no heterogeneity). \cite{lee2012kuhnian} studied heterogeneity in peer review, arguing that research commonly overlooks normatively appropriate disagreements among reviewers.  Additionally, Mallows-Binomial is unable to account for situations in which not all reviewers assess every proposal due to reviewer burden or conflicts of interest. The model that we develop in this paper can handle incomplete or partial rankings and estimates heterogeneous preferences among reviewers and the associated uncertainty, and therefore allows for accurate decision-making at the top of the list.

\vspace{-0.2in}
\paragraph{Setting 3: Modeling of Survey Preference Data under Heterogeneity} The third setting relates to the analysis of survey data, where survey respondents provide preferences on a collection of items. We study a survey dataset on the sushi preferences of Japanese adults \citep{kamishima2003nantonac}. Respondents provided a complete ranking of ten sushi types and rated them on a 5-point scale. The coarse rating scale leads to frequent ties between items. Furthermore, many respondents rated only a few sushi items, creating a substantial amount of missing data. Given the limited available data, we demonstrate how our proposed model accurately combines information from rankings and incomplete ratings to model preferences of respondents and identify heterogeneity.


In this paper, we propose a flexible, joint statistical model for rankings and ratings under heterogeneity: The Bradley-Terry-Luce-Binomial (BTL-Binomial).  Using a computationally-efficient Bayesian mixture of finite mixtures (MFM) of \cite{miller2018mixture}, we simultaneously estimate both the amount and type of heterogeneity among judges. We develop tools for model interpretation and goodness-of-fit assessment, and illustrate those on real and simulated datasets from the three motivating settings to demonstrate the value and practicality of analyzing preferences jointly with rankings and ratings. 

The rest of this paper is organized as follows. 
In Section \ref{sect:related} we provide notation and review background information on preference modeling and heterogeneity. In Section \ref{sect:BTLBMFM}, we describe the BTL-Binomial MFM approach for jointly modeling rankings and ratings under heterogeneous preference ideologies, derive an efficient Bayesian estimation procedure, and provide tools for model assessment. We use simulated and real data to illustrate the proposed model in Section \ref{sect:apps}. We conclude with a discussion in Section \ref{sect:discussion}.

\section{Notation and Background}\label{sect:related}

Suppose $I$ judges assess $J$ objects. Let $\mathcal{S}=\{1,\dots,J\}$ be the complete set of objects and $\mathcal{S}_i\subseteq\mathcal{S}$ be the subset assessed by judge $i$. Let $R_i = |\mathcal{S}_i|$ be the size of $\mathcal{S}_i$. Let $\Pi_i = \{\Pi_i(1)\prec\Pi_i(2)\prec\dots\prec\Pi_i(r_i)\}$ be judge $i$'s ranking of length $r_i\leq R_i$, such that $\Pi_i(r)$ is the $r^\text{th}$-most preferred object by judge $i$ among $\mathcal{S}_i$. Let $X_{ij}\in\{0,1,\dots,M\}$ be the rating of judge $i$ to object $j$, such that $0$ is the best and $M$ the worst. This reversed rating scale maintains a symmetry with rankings, in that numerically low ratings correspond to numerically low rankings.

\subsection{Statistical Ranking and Rating Models}

We model rankings via the Bradley-Terry-Luce (BTL) family of distributions, which includes the Bradley-Terry for paired comparisons and Plackett-Luce for groupwise comparisons, partial rankings, and complete rankings \citep{Bradley1952,Plackett1975,luce1959individual}. In these models, a \textit{worth} parameter, $\omega_j>0$, is assigned to each object $j$ such that larger values correspond to a higher probability of being ranked highly. The model assumes that the probability of drawing a specific ranking $\pi$ of length $R$ from a set $\mathcal{S}$ is,
\begin{equation}
    P[\Pi=\pi | \omega_1,\dots,\omega_J] = \prod_{r=1}^R \frac{\omega_{\pi(r)}}{\sum_{j\in\mathcal{S}} \omega_j - \sum_{s=1}^{r-1} \omega_{\pi(s)}}.
\end{equation}
The model may be interpreted as a sequential selection distribution where at each stage, the probability of selecting an unranked object is its worth divided by the summed worths of all unranked objects. Previous literature in political preferences, college rankings, and survey data, among others, has used the BTL family to model rankings \citep{Gormley2006,
Caron2014,mollica2017bayesian}.
The BTL family satisfies Luce's Choice Axiom \citep{luce1959individual}, which states that selecting one object over another should not be affected by the presence or absence of other objects. Furthermore, the BTL family can handle rankings in which some reviewers do not consider all objects. 

Statistical rating models are uncommon. Instead, it is often easier to use summary statistics such as the mean or median \citep{lee2013bias,tay2020beyond,NIHPeerReview}. When scores arise from a discrete, ordinal, finite, and equally-spaced set, various distributions may be appropriate after linear transformation to the space of integers. We follow \cite{pearce2022unified} and model ratings using a Binomial distribution due to its simplicity, unimodality, and natural mean parameter.

\subsection{Heterogeneous Preferences}

Standard preference models rely on the assumption that each object has a single, true underlying quality. This assumption is inappropriate when judges exhibit heterogeneity. For example, voters of different political parties may have diverging opinions of candidates. Another example arises in peer review, where reviewers may adhere to distinct ideologies for what constitutes promising research based on their background and training \citep{lee2012kuhnian}. In such situations, we say the judges exhibit heterogeneous preference ideologies.

A latent class mixture model can be used to capture heterogeneous preference ideologies. Mixture models have been used in the context of both Mallows \citep{Busse2007,Ali2010}
and BTL distributions \citep{Gormley2006,
mollica2017bayesian}. To the best of our knowledge, no joint statistical model for rankings and ratings under heterogeneity exists. Latent class preference models generally assume there exist $K$ preference ideologies and that each judge adheres to precisely one. Latent classes represent the preference ideologies, such that each class $k\in\{1,\dots,K\}$ has its own set of parameters. We let $Z_i=k$ denote judge $i$'s class. 

The true number of preference classes, $K$, is often unknown and must be identified or estimated. Most of the literature on heterogeneous preferences fits separate models
under various choices of $K$ and selects the best-fitting or most parsimonious model via some goodness-of-fit criteria \citep{Gormley2006,mollica2017bayesian}. However, it is also possible to estimate $K$ probabilistically. A vast Bayesian literature exists regarding these models (e.g., \cite{
nobile2004posterior}).
We elect to use a mixture of finite mixtures (MFM) approach. MFMs can be described as Bayesian latent class mixture models in which the number of classes itself is a random variable and assigned a prior. In their most general form, MFMs are easily interpretable and consistent for the true number of classes as the sample size grows \citep{miller2018mixture}.

\section{BTL-Binomial MFM Model}\label{sect:BTLBMFM}

We now propose a joint statistical model for rankings and ratings under heterogeneity.
Under the BTL-Binomial MFM model, the observed preference data $\Pi$ and $X$ are assumed to arise from the following generative model:
\spacingset{1.1}
\begin{equation}
\begin{aligned}
    K &\sim f_K(\cdot) & \text{$f_K$ is a pmf on $\{1,2,\dots\}$}\label{eq:MFM}\\
    \gamma &\sim f_\gamma(\cdot) & \text{$f_\gamma$ is a pdf on $\mathbb{R}^+$}\\
   \pi | K,\gamma &\sim \text{Dirichlet}_K(\gamma,\dots,\gamma)\\
    (p_k,\theta_k) &\overset{iid}{\sim} f_{p,\theta}(\cdot) & \text{$f_{p,\theta}$ is a pdf on $[0,1]^J\times\mathbb{R}^+$; } k\in\{1,\dots,K\}\\
    Z_i | \pi &\overset{iid}{\sim} \text{Categorical}(\pi_1,\dots,\pi_K) & i\in\{1,\dots,I\}\\
    \Pi_i, X_i|Z_i=k,p,\theta &\overset{ind.}\sim \text{BTL-Binomial}(p_k,\theta_k)
\end{aligned}
\end{equation}
\spacingset{1.5}

We briefly interpret the generative model. The \textit{number of heterogeneous preference ideologies}, $K$, is drawn from a prior. Independently, a \textit{concentration parameter} $\gamma>0$ is drawn from a hyperprior, where $\gamma$ controls the concentration of class weights between sparsity (few classes have substantial weight) and equality (all classes have equal weight). The \textit{class weights}, $\pi=(\pi_1,\dots,\pi_K)$ are then drawn from a symmetric Dirichlet prior. Given $K$,  \textit{class-specific preference parameters} $(p_k,\theta_k)$ are drawn from a prior for each class $k$. After drawing the class label $Z_i$ for each judge $i$, their ranking and ratings $(\Pi_i,X_i)$ are drawn from a BTL-Binomial distribution with class-specific parameters.

\subsection{BTL-Binomial Distribution}\label{sect:BTLB}

Suppose a judge assesses $J$ objects using a ranking, $\Pi$, and ratings, $X$. Assume $\Pi$ is of length $R\leq J$, and each rating $X_{j}$, $j\in\mathcal{S}$, is an integer between $0$ (best) and $M$ (worst). $\mathcal{S}$, $R$, and $M$ are fixed and known. Under a Bradley-Terry-Luce-Binomial (BTL-Binomial) distribution, their joint probability is given by:
\spacingset{1.1}
\begin{equation}
\begin{aligned}P[\Pi=\pi,X=x|p,\theta]&=\prod_{r=1}^R \frac{\exp(-\theta p_{\pi(r)})}{\sum\limits_{j\in\mathcal{S}} \exp(-\theta p_{j})-\sum\limits_{s=1}^{r-1} \exp(-\theta p_{\pi(s)})} \times \prod_{j=1}^J {M\choose x_j}p_j^{x_j}(1-p_j)^{M-x_j}\\
    p &= [p_1 \dots p_J]^T\in[0,1]^J, \theta>0,\\
    \Pi,X_1,&\dots,X_J\text{ are mutually independent}.
\end{aligned}
\label{eq:BTL-B}
\end{equation}
\spacingset{1.5}
The BTL-Binomial model in equation~(\ref{eq:BTL-B}) combines a BTL ranking distribution parameterized by worth parameters $\omega_j = \exp(-\theta p_j)$ and a Binomial rating distribution for each object, with Binomial probability $p_j$. We call $p$ the \textit{object quality vector}, which contains the underlying object qualities on the unit interval. The parameter $p$ appears in both ranking and rating components of the model and thus ties together their estimation to learn preferences. $\theta$ is the \textit{consensus scale parameter}, which measures of the strength of ranking consensus.

The Binomial rating parameterization is straightforward and follows \cite{pearce2022unified}. The BTL ranking parameterization that sets each $\omega_j$ to $\exp(-\theta p_j)$ is new and requires further explanation. Note that small values of $p_j$ correspond to large $\omega_j$, since a small-valued object quality parameter corresponds to a high-quality object, which should thus be ranked highly with greater probability (and vice versa). The parameterization maintains the exponential distance interpretation of the Mallows and Mallows-Binomial models, in that ranking probabilities are determined based on an exponential relationship with rate controlled by $\theta$ \citep{Fligner1986,pearce2022unified}. Here, the difference between the underlying qualities of two objects, $p_B-p_A$, is the distance that controls pairwise ranking probabilities. That is because,
\begin{align}
    P[A\prec B] =\frac{\exp(-\theta p_A)}{\exp(-\theta p_A)+\exp(-\theta p_B)} = \frac{1}{1+\exp(-\theta(p_B-p_A))}. \nonumber
\end{align}
The parameterization also removes the standard identifiability concern of BTL models because the worth parameters are now constrained to the interval $[\exp(-\theta),1]$ and anchored via the ratings. This claim is made formally in Theorem \ref{prop:btlb_ident}, whose proof is relegated to the Appendix.

\begin{theorem}\label{prop:btlb_ident}
Let $M$, $J$, and $R$ be fixed and positive integers such that $R\leq J$. Then the BTL-Binomial($p,\theta$) model is identifiable.
\end{theorem}

We now interpret BTL-Binomial parameters. The vector $p$ reflects object quality, in which values close to 0 (1) indicate high (low) quality.
$p$ may be ordered to form a consensus ranking, denoted $\pi_0$. For example, if $p=[0.5 \ 0.55 \ 0.1 \ 0.9]$ in a four object system, the consensus ranking $\pi_0 = 3\prec1\prec2\prec4$. Here, we say objects 1 and 2 are similar in quality, but object 3 is clearly highly quality than object 4. The consensus scale parameter $\theta$ is harder to interpret and may be considered a nuisance parameter. Most directly, $\theta$ is an input for calculating the probability that some object A is selected over object B in a pairwise tournament. For example, if $p_B-p_A=0.1$, then the probability that object $A$ is selected over $B$ is $1/(1+\exp(-0.1\times \theta))$ (see Table \ref{tab:theta}). Higher (lower) values of $\theta$ imply rankings among judges will be more (less) similar to each other.
\spacingset{1.1}
\begin{table}[h!!]
    \centering
    \begin{tabular}{ccccccc}
        $\theta$ & 1&5&10&20&40\\
        \hline
        $P[A\prec B|\theta, p_B-p_A=0.1]$ & 0.525 &  0.622  & 0.731  &   0.881 & 0.982
    \end{tabular}
    \caption{Pairwise ranking probabilities given $p_B-p_A=0.1$ under various $\theta$}\label{tab:theta}
\end{table}
\spacingset{1.5}

\subsection{Prior Selection}

Table \ref{tab:priors} summarizes model priors and hyperpriors. Following an example in \cite{fruhwirth2021generalized}, we assign $K$ a shifted Poisson prior such that $K-1\sim \text{Poisson}(\lambda)$. 
We assign $\gamma$ a Gamma$(\xi_1,\xi_2)$ hyperprior, as suggested in the Dirichlet Process Mixture (DPM) \citep{escobar1995bayesian}
and MFM \citep{miller2018mixture} literatures. When $\gamma$ is small, we expect some large, small, or even empty classes; when $\gamma$ is large the classes are expected to be roughly uniform in size. We assign $\pi$ (\textit{class weights}), a symmetric Dirichlet prior with concentration parameter $\gamma$. This so-called ``static" MFM is simple and common \citep{fruhwirth2021generalized}. We assign the BTL-Binomial parameters $p_{jk}$ i.i.d. $\text{Beta}(a,b)$ priors, which are not conjugate but simplify the posterior given Binomial ratings. Finally, we assign $\theta_k$ i.i.d. $\text{Gamma}(\gamma_1,\gamma_2)$ priors. 
\spacingset{1.1}
\begin{table}[h!!]
    \centering
    \begin{tabular}{ccc}
        Parameter & Interpretation & Prior\\
        \hline
        $K$ & Number of Ideology Classes &$\text{Poisson}(K-1|\lambda)$\\
        $\gamma$ & Dirichlet Concentration Parameter& Gamma$(\gamma|\xi_1,\xi_2)$\\
        $\pi$ & Class Weights & $\text{Dirichlet}_K(\pi|\gamma_,\dots,\gamma)$\\
        $p_k,\theta_k$ & BTL-Binomial Parameters &$\prod_{j=1}^J \text{Beta}(p_{jk}|a,b)\times\text{Gamma}(\theta_k|\gamma_1,\gamma_2)$  \\
    \end{tabular}
    \caption{Priors for the BTL-Binomial MFM Model}\label{tab:priors}
\end{table}
\spacingset{1.5}

Hyperparameter settings may be highly influential. $\lambda$ influences the prior expectation on $K$, such that $E_\lambda[K] = \lambda + 1$. However, the Dirichlet concentration parameter $\gamma$ allows for unequal weights between classes, thus influencing the number of non-empty classes, $K^+\leq K$. Values of $\gamma$ close to 0 allow for parsimony in the case of no heterogeneity; values greater than $1$ give higher probability to $K^+=K$. 
Selection of $\lambda$ is highly dependent on context; we suggest choosing the Gamma hyperparameters $\xi_1,\xi_2$ to provide density to $\gamma\in[0,3]$, which corresponds to substantial probability that $K^+$ may be any integer between $1$ and $K$. For the Beta hyperparameters, $a=b=1$ leads to a proper and minimally informative Uniform prior. Instead selecting $a$ and $b$ via an empirical Bayes approach based on the observed ratings may improve estimation efficiency. For the Gamma hyperparameters on $\theta$, setting $\gamma_1=1,\gamma_2=0$ leads to a flat but improper prior. We suggest choosing values to provide substantial density in the region $\theta\in[5,35]$, which corresponds to varying but reasonable levels of consensus.

\subsection{Estimation via Telescoping Sampler}\label{sect:bayes}

Until recently, MFM models have been computationally challenging to estimate due to difficulties associated with reversible jump MCMC (RJMCMC), the primary estimation tool \citep{nobile2004posterior}.
\cite{miller2018mixture} proved theoretical connections between the MFM and DPM models, thus expanding the toolkit and improving speed. Subsequently, \cite{fruhwirth2021generalized} proposed the ``telescoping sampler" which drastically lowered the computational burden of fitting MFM models. Their work cleverly decomposes the total number of latent classes, $K$, from the number of non-empty classes, $K^+$. Separating these quantities permits a simple Gibbs-type sampler (e.g., no RJMCMC) that is similar in form to those for Bayesian mixture models with fixed $K$. We adapt the telescoping sampler for the BTL-Binomial MFM model, which is presented in the Appendix alongside algorithms for estimation under fixed $K$ and maximum \textit{a posteriori} (MAP) estimation.

\subsection{Model Assessment}

We assess mixing and convergence by examining trace plots of quantities which are invariant to label-switching \citep{stephens2000dealing}. For MFM models, we follow the recommendation of \cite{fruhwirth2021generalized} to examine trace plots of $K^+$ or $\pi$ (ordered by mean). We also examine trace plots of $K$ and $\gamma$. If class-specific parameters suffer from label-switching, one may apply the algorithm of \cite{stephens2000dealing} to the posterior samples.

We also examine goodness of fit by comparing the observed and posterior predictive distributions of three types of statistics: (1) rating mean, by object; (2) rating variance, by object, and (3) pairwise probability that object $A$ is ranked above object $B$, for each pair of objects $(A,B)$. 
Under a well-fitting model, the observed and posterior predicted statistics should be similar. We assess similarity via visual inspection.


\section{Applications}\label{sect:apps}

We apply the BTL-Binomial model to three motivating examples: paper selection in large academic conferences under sparsity of comparisons, proposal selection in grant panel review under heterogeneity, and modeling of survey data under heterogeneity. 

\subsection{Paper Selection in Large Academic Conferences}\label{sect:appConf}
Our first application is to the paper selection process in large and highly competitive academic conferences. These conferences typically handle high volumes of paper submissions, and thus reviews are dispersed among many reviewers. Here, we simulate reviews in the form of ratings and rankings and use them to estimate proposal quality via the BTL-Binomial model. We focus on point estimation under the assumption of a single ideology among reviewers (i.e., fixed $K=1$). Beyond self-selection of papers into research areas and the timing restrictions of organizers, estimating heterogeneity in this context may be particularly noisy given the limited amount of data available from each reviewer.

\subsubsection{Simulation Setup}

Our simulation study is loosely based on that of \cite{liu2022integrating}, who proposed an algorithm to integrate rankings into ratings for 
the paper selection process used by the International Conference on Learning Representations 2017. 
In our study, we simulate a conference that has recruited $I=50$ reviewers to assess $J=50$ papers. Each reviewer $i$ cannot possibly assess every paper, so instead each is assigned a subset, $\mathcal{S}_i$, at random such that each paper receives an equal number of reviews. Specifically, $|\mathcal{S}_i|=R_i=R$ and each paper receives $R$ reviews since $I=J$. Reviewer $i$ first provides ratings $X_{ij}, j\in\mathcal{S}_i$. Ratings are integers between $0$ (exemplary) and $M$ (poor). Reviewers do not rate the other papers. Second, reviewer $i$ provides a top-4 ranking, $\Pi_i$, of their favorite papers among those assigned, without ties. We assume that a reviewer deems their ``unranked" papers (i.e., $\{j\in\mathcal{S}_i | j\not\in\Pi_i\}$) worse than those which were ranked. However, no information can be gleaned from reviewer $i$ for papers not in $\mathcal{S}_i$.

We generate ratings and top-4 rankings from a BTL-Binomial distribution. To capture different amounts of data and noise, we consider all combinations of the following values: (1) $R\in\{4,8,12,24\}$. Small $R$ signifies less work for each reviewer and provides less preference data. (2) $M\in\{4,9\}$ for a 5- or 10-point rating scale, respectively. Small $M$ increases the coarseness of ratings and thus the probability of ties. (3) $\theta\in\{1,10,20,40\}$. Large $\theta$ implies more consensus in rankings; see Table \ref{tab:theta}. In each simulation scenario, we draw $p_j\sim\text{Beta}(1,1) \overset{d}{=}\text{Uniform}[0,1]$, then fit a BTL-Binomial model with a single latent class to the data using hyperparameters $a=1,b=1,\gamma_1=5,\gamma_2=0.25$, which were chosen to be diffuse. Each scenario is replicated 100 times.

\subsubsection{Results}
We now demonstrate the model's ability to accurately estimate the true overall ranking of papers, $\pi_0$, and improve estimation of $\pi_0$ in comparison to the standard paper selection method based solely on ratings. That method is to order the papers by their mean rating and break ties on the basis of another reviewer \citep{shah2018design}. In the absence of additional reviewers, we use random tie-breaking. We let $\hat\pi_0^{\text{BTLB}}$ be the BTL-Binomial MAP estimate of $\pi_0$, determined by ordering papers based on $\hat p$. Similarly, we let $\hat\pi_0^\text{X}$ be the ratings-only MAP estimate of $\pi_0$, determined by ordering papers based on mean ratings. The accuracy of MAP estimates $(\hat p,\hat\theta)$ is shown in the Appendix. 

To measure the inaccuracy of each model, we calculate the percentage of object pairs in which the model incorrectly identifies the true order of the objects in $\pi_0$. This is equivalent to a normalized Kendall's $\tau$ distance between $\pi_0$ and a model estimate $\hat\pi_0$. We plot the mean inaccuracy across simulations for each combination of $R$, $M$, and $\theta$ from the BTL-Binomial and ratings-only models in Figure \ref{fig:simulation_decision}.

\spacingset{1.1}
\begin{figure}[h!!]
    \centering
    \includegraphics[width=\textwidth]{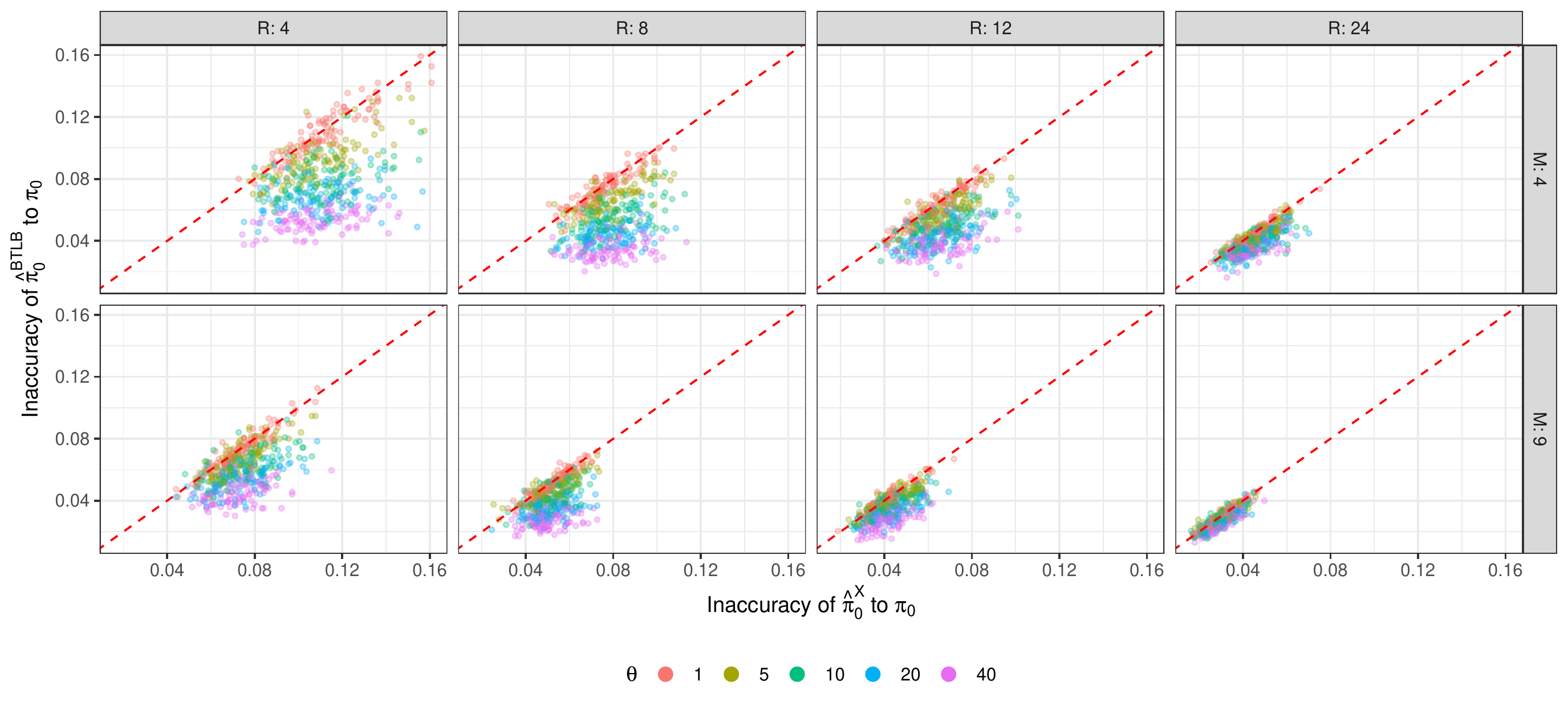}
    \caption{Scatterplots of the mean inaccuracy across 100 simulations of estimated $\hat\pi_0^{BTLB}$ (BTL-Binomial model) and $\hat\pi_0^\text{X}$ (standard ratings-only model) to the true ranking of papers $\pi_0$ under combinations of $R$, $M$, and $\theta$.}
    \label{fig:simulation_decision}
\end{figure}
\spacingset{1.5}

The BTL-Binomial model outperforms the standard ratings-only model on average for every combination of $M$, $R$, and $\theta$. The largest improvement in estimation accuracy with the BTL-Binomial over the standard ratings-only model occurs when $R$ and/or $M$ is small, which are precisely the settings of the utmost interest for large academic conferences. These results should be intuitive: When $R$ is small, each paper receives few assessments and thus the additional information from rankings is highly beneficial to accurate preference modeling. When $M$ is small, ties will be common in ratings and lead to haphazard estimation based on random tie-breaking; rankings help to accurately break those ties. We also notice that as $\theta$ increases, so does the accuracy of the BTL-Binomial model. This is because higher $\theta$ means that rankings will be more adherent to the true ranking $\pi_0$ on average, and thus will provide less noisy information for accurate modeling of paper quality. 

We have demonstrated that the BTL-Binomial model leads to more accurate decision-making using rankings and ratings in academic conference paper selection under realistic review settings. A key benefit is that reviewers need not assess many papers or greatly increase their workload. In fact, even with a coarse 5-point rating scale and top-4 rankings of a small number of papers, quality assessments may be made based on rankings and ratings with greater accuracy in comparison to the standard mean-ratings model. Furthermore, there becomes little need for random or subjective tie-breaking, making the work for data aggregators and conference chairs both easier and more objective.

\subsection{Proposal Selection in Grant Panel Review under Heterogeneity}\label{sect:appPanel}

Our second application is to a grant panel review administered by the American Institute of Biological Sciences (AIBS) during the 2021 season \citep{Gallo2023}. The AIBS issues a call for funding, recruits a panel of qualified reviewers, and administers the peer review process. Prior to panel discussion, reviewers are given access to the grant proposals, although do not necessarily read each of them in detail. During panel discussion, each proposal is discussed in turn and each reviewer provides a rating which reflects the overall scientific merit of each proposal using the numbers between 1 (\textit{excellent}) and 5 (\textit{poor}) in single decimal point increments (which we transform to the integers between 0 and $M=40$). After discussion, each reviewer provides a top-6 ranking of their overall preferred proposals. The ranking is not required to align with the reviewer's ratings. The AIBS would like to know if there are distinct preference groups among reviewers, what those preferences are, and know how much uncertainty exists in the estimated proposal quality assessments.

Some rankings and ratings are missing due to conflicts of interest and other reasons unrelated to proposal quality (e.g., intermittent distractions, internet connectivity issues, lack of qualification to accurately review). For missing ratings, we simply remove the corresponding Binomial components from the likelihood. If a reviewer does not provide a ranking, we remove the corresponding BTL components from the likelihood. If reviewer $i$ does not rank proposals due to a conflict of interest, these proposals are removed from his/her set of proposals $\mathcal{S}_i$. On the basis of Luce's Choice Axiom, estimation of model parameters corresponding to proposals with which a reviewer has a conflict of interest is not affected.

\subsubsection{Exploratory Analyses}

We study a panel with $I=17$ reviewers and $J=25$ proposals. Figure \ref{fig:aibs_EDA} displays boxplots of ratings and bar charts of rankings, by proposal. 
The mean and variance of ratings given to each proposal highly vary. For rankings, 14 of the 25 proposals are included in at least one judge's top-6 ranking.
\spacingset{1.1}
\begin{figure}[h!!]
    \centering
    \includegraphics[width=\textwidth]{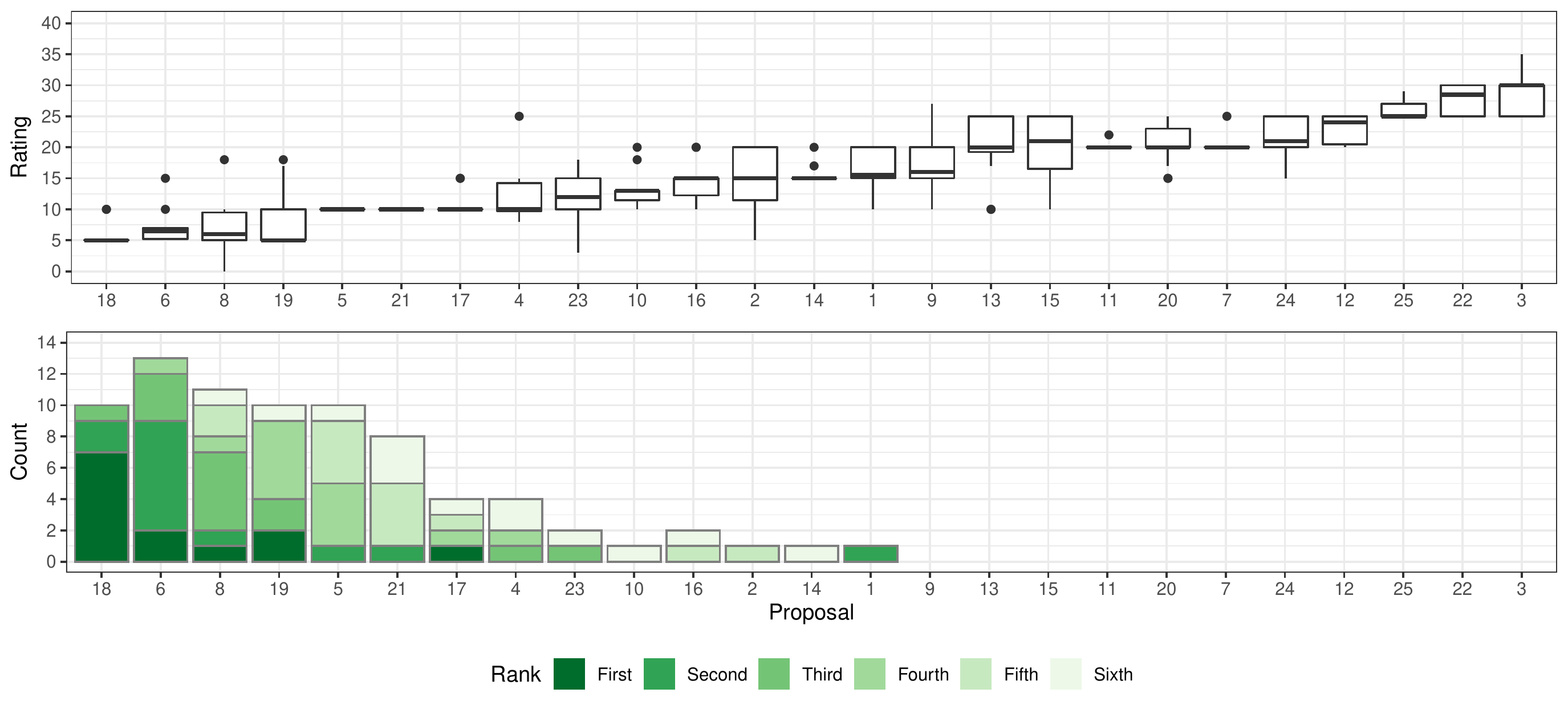}
    \caption{Boxplots of ratings (\textit{top}) and stacked bar charts of ranks (\textit{bottom}) by proposal.}
    \label{fig:aibs_EDA}
\end{figure}
\spacingset{1.5}

Discrepancies between quality assessments via ratings and rankings by the same judge were common: only one judge provided internally consistent ratings and rankings, and two judges (8 and 12) provided highly inconsistent ratings and rankings. This aligns with psychological research that the two assessments rely on different cognitive processes \citep{goffin2011all}. This suggests that rankings can provide additional and unique information.

\subsubsection{Estimation and Results}

We now fit a BTL-Binomial MFM model to the AIBS data. We set priors as follows: Given the small sample size, we choose $\lambda = 1$, $\xi_1=2$, and $\xi_2=3$ to assign prior weight primarily to 
$K^+\in\{1,2,3\}$.
We set $a=2.50$ and $b=3.77$ using an empirical Bayes approach. 
We set $\gamma_1=10$ and $\gamma_2=0.5$ to provide substantial weight to values of $\theta\in[5,35]$. 
Further information on model estimation and assessment is provided in the Appendix.

Figure \ref{fig:aibs_results} displays model results. 
The top-left panel displays the posterior of $K^+$, which provides strong evidence for a 2-class model.  Thus, we display results conditional on $K^+=2$ in what follows. 
The top-right panel displays estimated class membership probabilities by judge. Class 1 (red) includes 16 judges and class 2 (blue) includes just one. The bottom panel displays the posteriors of preference parameters for each proposal and class. Class 1 prefers proposals 18, 6, 8, and 19, in that order, and exhibits relatively strong consensus. Alternatively, class 2 largely reflects judge 8. The high levels of uncertainty in class 2's preference parameters reflects that it is comprised of a single reviewer.
Given the ``outlier'' judge,
the funding agency may wish to consider what made judge 8 provide such unique preferences and decide if those warrant separate consideration, or if the results from class 1 should be considered alone in making funding decisions.
\spacingset{1.1}
\begin{figure}[h!!]
    \centering
    \includegraphics[width=\textwidth]{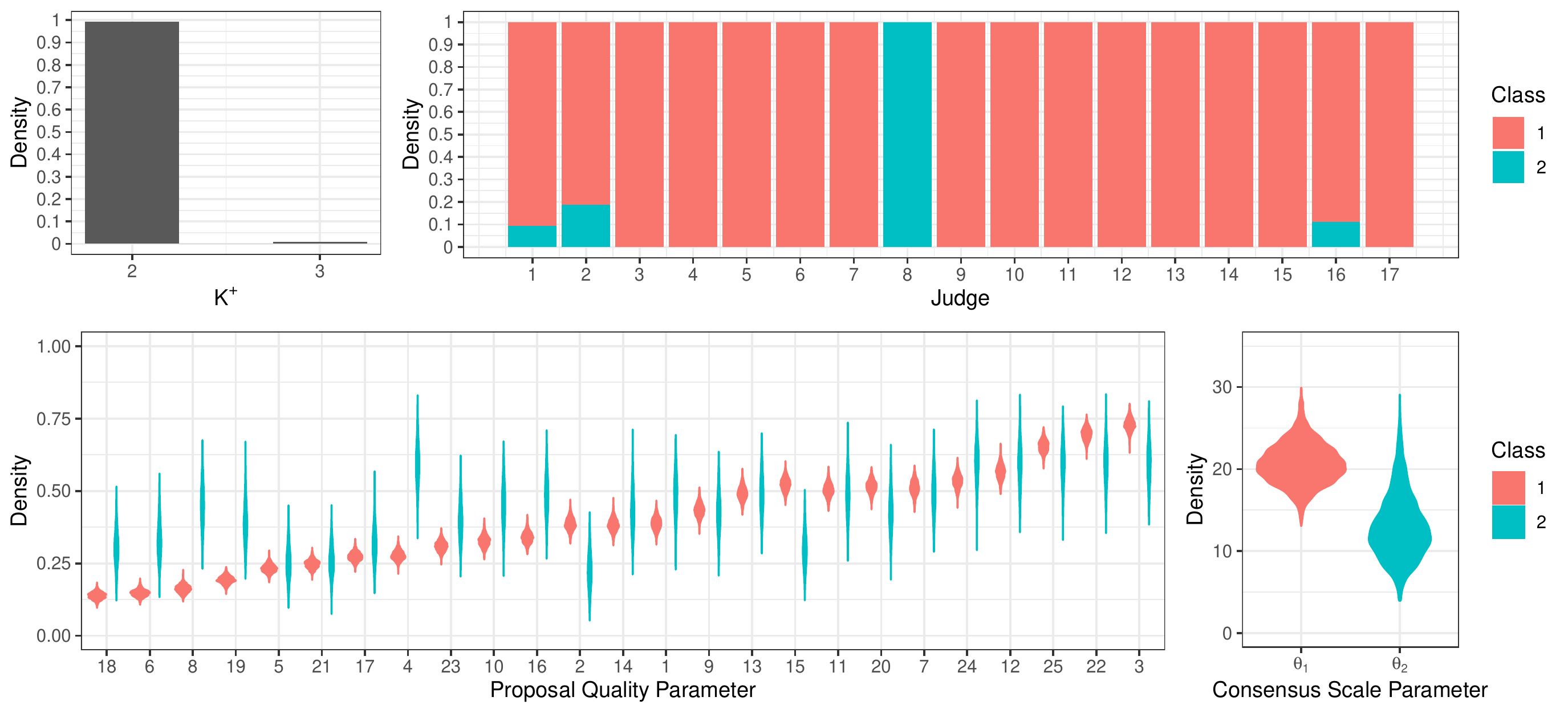}
    \caption{Posterior summaries of $K^+$ (\textit{top-left}); class membership probabilities given $K^+=2$ (\textit{top-right}); and class-specific preference parameters given $K^+=2$ (\textit{bottom}).}
    \label{fig:aibs_results}
\end{figure}
\spacingset{1.5}

\subsection{Modeling Survey Preference Data under Heterogeneity}\label{sect:appSushi}

Our final application is to survey data on the sushi preferences of $I=5,000$ Japanese adults \citep{kamishima2003nantonac}. Both rankings and ratings were collected as part of the survey. Respondents were first asked to rank a collection of $J=10$ sushi types from best to worst. The sushi types were fatty tuna, tuna, shrimp, tuna roll, sea eel, salmon roe, squid, egg, sea urchin, and cucumber roll. Each respondent was shown the same collection of sushi types and provided a complete ranking. Second, the respondents were asked to rate the sushi types using a 5-point integer scale at will, coded from 0 (best) to $M=4$ (worst). Each respondent generally provided only a few ratings. Our goal is to probabilistically model the amount and type of heterogeneity in sushi preferences among survey respondents.

The present survey dataset is uniquely positioned for analysis via a joint ranking and rating model: While rankings are complete, they lack granularity; the ratings provide granularity but have a very high rate of missingness. For the purpose of understanding heterogeneity among judges, using both rankings and ratings may be especially helpful for accurate inference on preferences.

\subsubsection{Exploratory Analyses}

Figure \ref{fig:sushi_eda} displays ratings (top) and rankings (bottom) by sushi type. For ratings, most sushi types have unimodal distributions with right skew. However, cucumber roll has a unimodal distribution centered at the middle rating while sea urchin has a bimodal distribution with peaks at the best and worst ratings. We can also see relative differences in the number of ratings each sushi type received. For example, the lower density of points for tuna roll and cucumber roll implies fewer respondents rated these types. 74.22\% of ratings are missing. For rankings, fatty tuna was ranked first by approximately one-third of respondents, while cucumber roll was ranked last by approximately one-third. Ties were not allowed in rankings. Consequently, we observe more demarcation in ranking distributions than in ratings. For example, fatty tuna and tuna have similar rating distributions but far more respondents ranked fatty tuna in first place than tuna.

\spacingset{1.1}
\begin{figure}[h!]
    \centering
    \includegraphics[width=\textwidth,height=300pt]{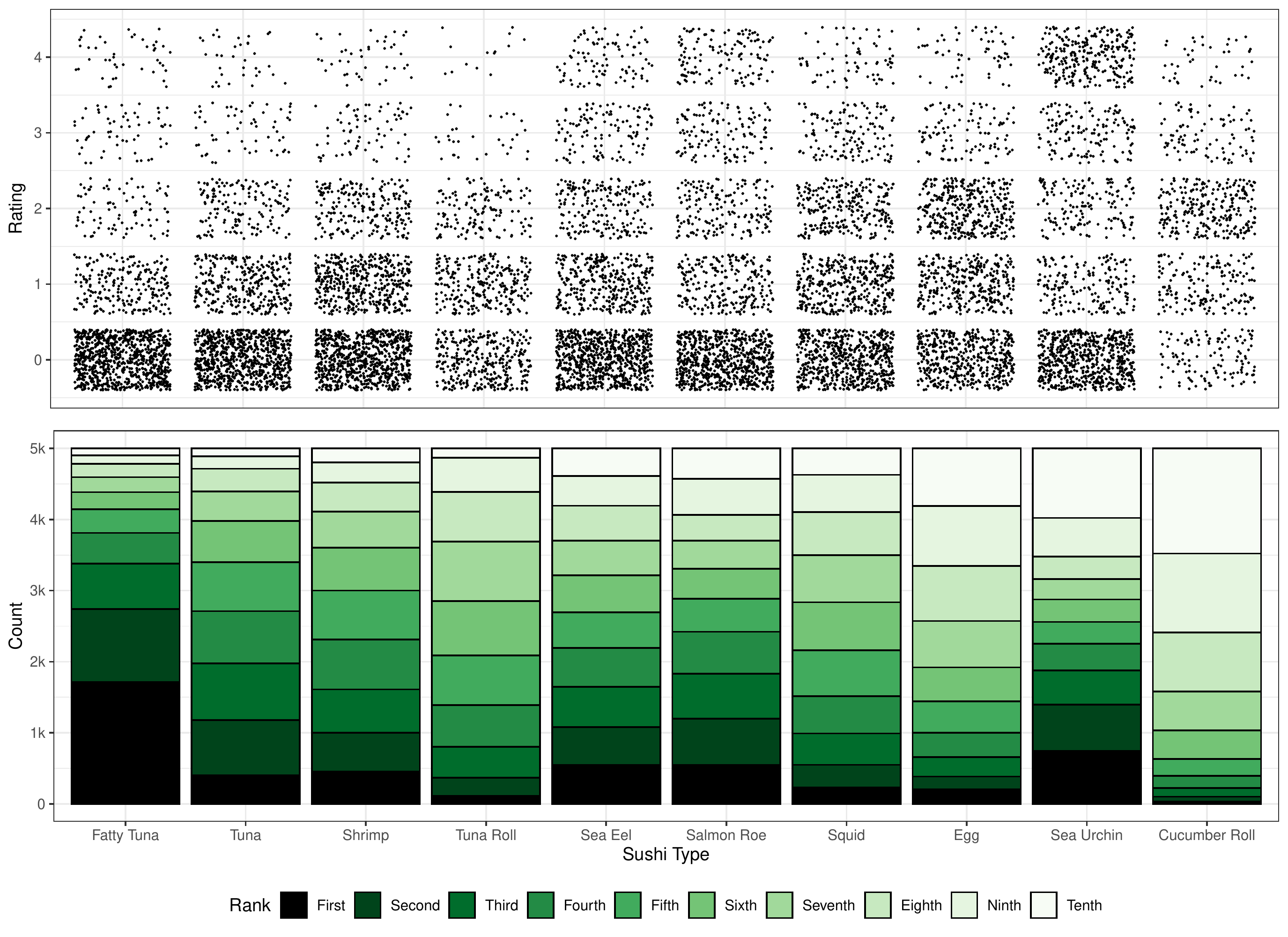}
    \caption{Ratings (\textit{top}) and stacked bar charts of rankings (\textit{bottom}) by sushi type.}
    \label{fig:sushi_eda}
\end{figure}
\spacingset{1.5}


\subsubsection{Estimation and Results}

We now fit a BTL-Binomial MFM model to the sushi data. We set priors as follows: To aid interpretability given the large sample size, we assign prior weight to moderate number of classes using $\lambda = 7$, $\xi_1=2$, and $\xi_2=3$. 
We set $a=0.26$ and $b=0.77$ using an empirical Bayes approach. We set $\gamma_1=20$ and $\gamma_2=1$ to provide substantial weight to $\theta\in[10,30]$. Further information on model estimation and assessment is provided in the Appendix.

Results show high posterior probability on $K^+=8$, indicating that 8 heterogeneous preference classes exist among the respondents.
Table \ref{tab:sushi} summarizes the estimated classes conditional on $K^+=8$. Each row contains the posterior mean population proportion $\hat\pi_k$, top-3 sushi preferences, and posterior mean consensus scale parameter $\hat\theta_k$ for each estimated class (rows ordered by $\hat\pi_k$). Classes 1 and 2 are the largest in size and exhibit the highest consensus. Thus, these classes may reflect reasonably homogeneous plurality classes. The remaining classes represent smaller proportions of the survey respondents and exhibit relatively weak consensus. Still, we may think of these classes as representing subgroups in the population that are present but less well-defined.
\spacingset{1.1}
\begin{table}[ht]
\centering
\begin{tabular}{cclc}
Class & $\hat\pi_k$ & Top-3 Sushi Preferences & $\hat\theta_k$ \\ 
  \hline
  1 & 0.20 & Sea Urchin, Fatty Tuna, Salmon Roe & 10.71 \\ 
    2 & 0.17 & Fatty Tuna, Tuna, Shrimp & 15.11 \\ 
    3 & 0.16 & Sea Eel, Fatty Tuna, Sea Urchin & 6.76 \\ 
    4 & 0.12 & Fatty Tuna, Tuna, Sea Eel & 4.39 \\ 
    5 & 0.10 & Salmon Roe, Sea Eel, Shrimp & 4.20 \\ 
    6 & 0.10 & Salmon Roe, Fatty Tuna, Tuna & 4.33 \\ 
    7 & 0.10 & Squid, Shrimp, Tuna & 3.47 \\ 
    8 & 0.05 & Egg, Shrimp, Cucumber Roll & 2.77
\end{tabular}
\caption{Posterior summaries of preference classes conditional on $K^+=8$} \label{tab:sushi}
\end{table}
\spacingset{1.5}

\section{Discussion}\label{sect:discussion}

In this paper, we propose a statistical model for joint analysis of rankings and ratings under heterogeneity, the BTL-Binomial MFM model. The model is quite flexible in several important ways. First, it  allows for analyzing ranking or preference data of various types -- pairwise comparisons, partial rankings, as well as complete rankings -- jointly with ratings. Second, the model allows for incomplete designs (or structural missingness) where each judge by design would only be assigned to review certain objects and not others. Such ``separate ballots" could arise from either conflicts of interest or selective assignments to reflect expertise or manage judges' workload. The model can also accommodate missing at random data where rankings or ratings are missing due to circumstances unrelated to the quality of objects being assessed, such as reviewer fatigue or cases when subsets of reviewers only rate or only rank the objects.
Third, the model allows for reviewer rankings to be inconsistent with their own ratings, which often happens in practice. 
Fourth, BTL-Binomial formulation satisfies Luce's Choice Axiom and the related independence from irrelevant alternatives criterion  which is a desirable property in social choice theory \citep{arrow1950difficulty}.
Fifth, the model 
simultaneously estimates both the number of heterogeneous ideologies and the specific preferences of each group, as well as the associated uncertainty. 

The BTL-Binomial MFM model makes few parametric assumptions on the reviewers, objects, and data. The model assumes that, for reviewers in each ideological class, proposals have a true underlying quality that can be measured on the unit interval. Rankings and ratings must reflect random deviations from the assumed truth. Rankings must arise from the Bradley-Terry-Luce (BTL) family of ranking distributions that
allows for a variety of ordinal data types and has been used in a large variety of application areas (see Section \ref{sect:related}). On the other hand, ratings are assumed to be Binomial. This assumption is appropriate for integer-valued ratings arising from  
an ordinal and equally-spaced set with minimum and maximum allowable values. The mean-variance relationship imposed by the Binomial can be tested for validity after model estimation (e.g., see the Appendix). If this assumption is not met, it may be indicative of an incorrectly estimated number of heterogeneous preference groups.  We have proposed sensible goodness-of-fit criteria for assessing the parametric assumptions imposed by the model.

We fit the model by adapting the telescoping sampler of \cite{fruhwirth2021generalized}, which provides computationally efficient Bayesian estimation and a natural approach to uncertainty quantification. Furthermore, we provide algorithms for posterior sampling and MAP estimation under a fixed and pre-specified number of heterogeneous preference ideology classes, $K$.

The MFM approach was chosen to estimate heterogeneity in preferences among judges. We find the Bayesian framework of MFM models to be attractive for three principal reasons. First, Bayesian estimation allows for the incorporation of prior knowledge into the estimation procedure. This is useful in the case of limited preference data, which is common to many applications. When no prior knowledge is available, flat and/or minimally informative priors are available for all model parameters. Second, Bayesian estimation provides a unifying framework for obtaining uncertainty estimates, which is a key component of our work. Third, the proposed estimation procedure may actually reduce computation time when compared to frequentist procedures, due to the fact that $K$ is estimated simultaneously with model parameters and therefore removes the need to repeatedly fit models with different values of $K$. Additionally, analytic uncertainty results are unavailable for the BTL-Binomial model (and the related Mallows-Binomial model) in the frequentist setting and therefore require the bootstrap, which can be extremely computationally burdensome \citep{pearce2022unified}.
This is avoided by the present Bayesian approach where uncertainty estimation is natural. We note that another reasonable approach would have been to use a Dirichlet Process Mixture \citep{escobar1995bayesian}. However, it may be computationally difficult to estimate in high dimensions due to estimation via reversible jump MCMC and is inconsistent for the true number of latent classes \citep{miller2018mixture}. Given that we are interested not only in density estimation, but also in the latent classes themselves, the MFM approach is more suitable.

Two recent models for rankings and ratings may be directly compared to the BTL-Binomial MFM. First, the Mallows-Binomial is a joint statistical model for rankings and ratings that combines a Mallows ranking distribution with independent Binomial rating distributions \citep{pearce2022unified}. Unlike the BTL-Binomial MFM, the Mallows-Binomial allows only for partial or complete rankings, does not satisfy Luce's Choice Axiom, cannot easily handle separate ballots, and cannot estimate heterogeneity directly.
Furthermore, Mallows-Binomial is estimated in a frequentist framework, which is computationally slow when the number of objects is large or when uncertainty estimates are desired, which requires the bootstrap. As a result, the BTL-Binomial MFM is much more flexible while still providing a unified and statistical approach to preference modeling with rankings and ratings. A second comparable work is \cite{liu2022integrating}, who propose a non-parametric algorithm for integrating rankings into ratings. That algorithm is not statistical and does not yield a preference ordering, but instead returns a ``de-quantized" score for each judge and object in a fully data-driven approach. De-quantized scores may be useful and practical for decision-making, but they do not allow for estimation of (heterogeneous) preferences or their inherent uncertainty. \cite{liu2022integrating} also assume internal consistency between rankings and ratings, which is not practical in the motivational settings described herein.

To demonstrate the utility and benefits of modeling preferences with both rankings and ratings, we analyzed three datasets motivated by real-world settings. The first was a paper selection dataset from a hypothetical large and highly competitive academic conference, inspired by current paper review procedures in computer science conferences. Using these simulated data, we showed that incorporating rankings into the traditional rating system improves accuracy of paper selection and reduces the frequency of ties among estimated paper qualities.  These benefits exist even when a coarse 5-point rating scale is employed, and are especially apparent when each reviewer can only assess a few papers. While our simulation studies show that similar benefits can be achieved by increasing the number of papers assessed by each judge, our proposal of incorporating top-4 rankings into the analysis achieves the same benefits with little additional cognitive burden on reviewers. Thus, our analyses show that collecting top-4 rankings and using the BTL-Binomial model for estimating paper quality will make the work of paper selection by computer science conference area chairs both easier and more objective. 

The second example concerned a smaller-scale case of peer review, in which a panel assessed grant proposals using rankings and ratings. In this setting, the ability of the BTL-Binomial MFM model to flexibly handle a variety of realistic complexities was demonstrated. These include missing rankings and ratings due to conflicts of interest, reviewer fatigue, and logistical difficulties during the review panel; inconsistency between ratings and rankings at the reviewer level; and potential heterogeneity in preferences among reviewers.
We demonstrated how the BTL-Binomial MFM model naturally handles missing data and inconsistent rankings and ratings through a flexible model formulation and simultaneously estimates the number of heterogeneous preference groups among the reviewers with the overall preferences and level of consensus in each group. We identified two heterogeneous preference groups: a dominant collection of reviewers and an ``outlier" reviewer whose opinions may have otherwise unduly influenced the panel decision. However, our model could capture different types of heterogeneity, such as potential reviewer preferences for basic versus translational science that have been noted previously \citep{lee2013bias,smith2021reimagining}.
Model results may be used to communicate uncertainty in peer review quality assessment which is important for funding decisions; 
see \cite{gallo2022new} for an illustration on another AIBS panel review dataset.

The third example relates to the analysis of survey preference data, in which Japanese adults were asked to rate and rank common sushi types. In this setting, the model successfully combines rankings (which are complete but lack granularity) with ratings (which provide granularity but with a very high rate of missingness). Previous work on preference surveys has modeled heterogeneity among respondents using rankings (e.g., \cite{
Wang2017}) and ratings (e.g., \cite{patterson2002latent}). 
However, the number of mixture components is usually selected in advance or via goodness-of-fit statistics.
Here, we estimate the number of heterogeneous preference groups, as well as their population proportion, preferences, and level of consensus concurrently.

Additional research is needed on joint models for rankings and ratings. As noted previously, the specific parametric form of the rating model based on the Binomial distribution implies that ratings arise from a discrete, ordinal, finite, and equally-spaced set, and furthermore imposes a specific mean-variance relationship on the ratings. These assumptions may not be valid in some contexts, and extensions of our model to modify or relax these parametric assumptions may be useful. Furthermore, the relative influence of rankings and ratings in this model depends on the amount of available data of each type. In some cases, the ability to weight the importance of rankings and ratings during estimation may be important to some practitioners, particularly in contexts where either rankings or ratings are thought to be more relevant. In addition, the BTL-Binomial model also does not incorporate covariates or predictors, which may be of interest. The model can be extended to include covariates in a similar fashion as in BTL models \citep{tkachenko2016plackett}.

Incorporating rankings into existing decision-making processes or analyses that currently use only ratings, or vice versa, has certain benefits. From a psychological or psychometric perspective, rankings force demarcation and make explicit comparisons but are coarse and impose high cognitive load on the judges. On the other hand, ratings may provide granularity and allow for ties yet may be highly subjective or inconsistent. The BTL-Binomial MFM model provides a principled Bayesian approach for analyzing various types of ranking and rating data jointly, can account for the separate ballots problem, allows for data missing at random, and imposes minimal parametric assumptions. Furthermore, the model estimates the amount and type of heterogeneity among reviewers concurrently. Examples from three different contexts demonstrate practical applicability of this model for learning preferences. For large-scale conference review, the model provides a mechanism for tie-breaking similar-quality proposals without requiring reviewers to consider more than a few papers. In small-scale panel review, the model successfully identifies an ``outlier reviewer" and estimates the preferences of the dominant subgroup for decision-making. In survey sushi data, the model estimates heterogeneous groups of respondents with their distinct preferences, even in the presence of substantial missingness. Overall, we find the BTL-Binomial MFM to be useful and efficient in estimating heterogeneous preferences from rankings and ratings jointly.

\spacingset{1.1}
\bibliography{main}
\clearpage
\appendix
\section{Estimation Algorithms}

\subsection{BTL-Binomial MFM}
\subsubsection{Estimation Algorithm}
We now present an algorithm to estimate the BTL-Binomial MFM model based on the telescoping sampler of~\cite{fruhwirth2021generalized}. A few additional details are provided below the statement of the algorithm.

\bigbreak

\noindent \underline{\textbf{Algorithm 1: Telescoping Sampler for BTL-Binomial MFM Model}}\\
\noindent\textbf{Data:} $\Pi,X$\\
\noindent\textbf{Hyperparameters:} $\lambda,\xi_1,\xi_2,a,b,\gamma_1,\gamma_2$\\
\noindent\textbf{Parameters:} $B^\text{Gibbs},B^\text{MH},K_\text{start},\sigma^2_p,\sigma^2_\theta, \sigma^2_\gamma$

\begin{enumerate}
    \item Initialize parameters $\gamma$, $\pi_k$, $p_k$, and $\theta_k$ for $k\in\{1,\dots,K_\text{start}\}$.
    \item Repeat $B_\text{Gibbs}$ times:
    \begin{enumerate}
        \item \textit{Update $Z$ and $K^+$}:
        \begin{enumerate}
            \item Update $Z_i$, $i=1,\dots,I$, using 
            \begin{align*}
                P[Z_i=k|\pi,p,\theta,X_i,\Pi_i] &=\frac{ P[Z_i=k]P[X_i,\Pi_i|Z_i=k,p_k,\theta_k]}{\sum_{k'=1}^K P[Z_i=k']P[X_i,\Pi_i|Z_i=k',p_{k'},\theta_{k'}]}\\
                &=\frac{ \pi_kP[X_i,\Pi_i|Z_i=k,p_k,\theta_k]}{\sum_{k'=1}^K \pi_{k'}P[X_i,\Pi_i|Z_i=k',p_{k'},\theta_{k'}]}.
            \end{align*}
            \item Calculate $N_k = \sum_{i} I\{Z_i=k\}$ and $K^+ = \sum_{k=1}^K I\{N_k>0\}$, where $I\{\cdot\}$ is the indicator function. Relabel the classes such that the first $K^+$ are non-empty.
        \end{enumerate}
        \item \textit{Update (Non-Empty) Class Parameters}: For $k=1,\dots,K^+$, repeat $B^\text{MH}$ times: 
        \begin{enumerate}
            \item Sample each $p_{jk}$, $j=1,\dots,J$, via random-walk Metropolis-Hastings with proposal distribution $\text{Normal}(p_{jk},\sigma^2_p)$.
            \item Sample $\theta_k$ via random-walk Metropolis-Hastings with proposal distribution $\text{Normal}(\theta_{k},\sigma^2_\theta)$.
        \end{enumerate}
        \item \textit{Update $K$ and $\gamma$}:
        \begin{enumerate}
            \item Sample directly from $K|K^+,\gamma$ such that $K\in\{K^+,K^++1,\dots\}$ using
            $$P[K|K^+,\gamma] \propto f_K(K)\frac{K!}{(K-K^+)!}\frac{\Gamma(\gamma K)}{\Gamma(I+\gamma K)}, \ \ \ K=K^+,K^++1,\dots$$
            \item Sample from $\gamma|Z,K$ via random-walk Metropolis-Hastings with proposal distribution $\text{Normal}(\gamma,\sigma^2_\gamma)$ using the unnormalized density
    $$P[\gamma|Z,K] \propto f_\gamma(\gamma)\frac{\Gamma(\gamma K)}{\Gamma(I + \gamma K)}\prod_{k=1}^{K^+} \frac{\Gamma(N_k+\gamma)}{\Gamma(\gamma)}$$
    (\cite{fruhwirth2021generalized}, Algorithm 3).
        \end{enumerate}

    and sample $\gamma|Z,K$ via random walk Metropolis-Hastings

        \item \textit{Update Empty Classes and $\pi$}:
        \begin{enumerate}
            \item If $K>K^+$, sample $(p_k,\theta_k)$ directly from its prior for $k=K^++1,\dots,K$.
            \item Sample $\pi|K,\gamma,Z\propto\text{Dirichlet}(\gamma+N_1,\dots,\gamma+N_K)$.
        \end{enumerate}
    \end{enumerate}
\end{enumerate}

\subsubsection{Further Details of Algorithm 1}

\underline{Selection of Algorithmic Parameters:} $B^\text{Gibbs}$ is the number of times steps 2(a)-(d) will be repeated, which corresponds to the number of unique times $Z$, $K^+$, $K$, $\gamma$, and $\pi$ are re-sampled. $B^\text{MH}$ is the number of times within each Gibbs iteration that each set of class parameters $(p_k,\theta_k)$ will be re-sampled using Metropolis-Hastings. Setting $B^\text{MH}>1$ often improves algorithm efficiency since class labels tend to stabilize over the course of a chain and hence need not be resampled each time that the class parameters are updated. $K_\text{start}$ is the initial number of classes, which may be chosen to be any integer between 1 and $I$. We find that setting $K_\text{start}=I$ often yields quickly converging chains but comes with a computational cost when $I$ is large in the first few Gibbs iterations as classes collapse. Lastly, the Metropolis-Hastings proposal variance parameters $\sigma^2_p,\sigma^2_\theta$, and $ \sigma^2_\gamma$ should be chosen such that the acceptance probabilities are reasonably efficient (see \cite{gelman2013bayesian} for further information on the efficiency of Metropolis-Hastings). Tuning these parameters after a short initial test run can be useful.

\underline{Initialization:} One may initialize $\gamma$, $\pi$, $p$, and $\theta$ by sampling these quantities from their prior distributions. For greater efficiency, one may also initialize using MAP estimates for a prespecified choice of $K$.

\subsection{BTL-Binomial Estimation under Fixed $K$}

\subsubsection{Model Statement}

Suppose the following BTL-Binomial latent class mixture model, under fixed $K\geq1$:

\begin{equation}
\begin{aligned}
    \gamma &\sim f_\gamma(\cdot) & \text{$f_\gamma$ is a p.d.f. on $\mathbb{R}^+$}\label{eq:FixedK}\\
   \pi | \gamma &\sim \text{Dirichlet}_K(\gamma,\dots,\gamma)\\
    (p_k,\theta_k) &\overset{iid}{\sim} f_{p,\theta}(\cdot) & \text{$f_{p,\theta}$ is a p.d.f. on $[0,1]^J\times\mathbb{R}^+$; } k\in\{1,\dots,K\}\\
    Z_i | \pi &\overset{iid}{\sim} \text{Categorical}(\pi_1,\dots,\pi_K) & i\in\{1,\dots,I\}\\
    \Pi_i, X_i|Z_i=k,p,\theta &\overset{ind.}\sim \text{BTL-Binomial}(p_k,\theta_k)
\end{aligned}
\end{equation}
Note that the above generative model in \ref{eq:FixedK} is identical to the BTL-Binomial MFM model, less the initial sampling of $K$.  Furthermore, we assume the same priors $f_\gamma$ and $f_{p,theta}$.

\subsubsection{Estimation Algorithm}
We present an algorithm to estimate a BTL-Binomial latent class mixture model with fixed $K$.

\noindent \underline{\textbf{Algorithm 2: Sampler for BTL-Binomial Latent Class Mixture Model}}\\
\noindent\textbf{Data:} $\Pi,X,K$\\
\noindent\textbf{Hyperparameters:} $\xi_1,\xi_2,a,b,\gamma_1,\gamma_2$\\
\noindent\textbf{Parameters:} $B^\text{Gibbs},B^\text{MH},\sigma^2_p,\sigma^2_\theta, \sigma^2_\gamma$

\begin{enumerate}
    \item Initialize parameters $\gamma$, $\pi_k$, $p_k$, and $\theta_k$ for $k\in\{1,\dots,K\}$.
    \item Repeat $B_\text{Gibbs}$ times:
    \begin{enumerate}
        \item \textit{Update $Z$}:
        \begin{enumerate}
            \item Update $Z_i$, $i=1,\dots,I$, using 
            $$P[Z_i=k|\pi,p,\theta,X_i,\Pi_i] =\frac{ \pi_kP[X_i,\Pi_i|Z_i=k,p_k,\theta_k]}{\sum_{k'=1}^K \pi_{k'}P[X_i,\Pi_i|Z_i=k',p_{k'},\theta_{k'}]}$$
            \item Calculate $N_k = \sum_{i} I\{Z_i=k\}$, where $I\{\cdot\}$ is the indicator function.
        \end{enumerate}
        \item \textit{Update Class Parameters}: For $k=1,\dots,K$, repeat $B^\text{MH}$ times: 
        \begin{enumerate}
            \item Sample each $p_{jk}$, $j=1,\dots,J$, via random-walk Metropolis-Hastings with proposal distribution $\text{Normal}(p_{jk},\sigma^2_p)$.
            \item Sample $\theta_k$ via random-walk Metropolis-Hastings with proposal distribution $\text{Normal}(\theta_{k},\sigma^2_\theta)$.
        \end{enumerate}
        \item \textit{Update $\gamma$}: Sample from $\gamma|Z,K$ via random walk Metropolis-Hastings using the unnormalized density
            $$P[\gamma|Z] \propto f_\gamma(\gamma)\frac{\Gamma(\gamma K)}{\Gamma(I + \gamma K)}\prod_{k=1}^{K} \frac{\Gamma(N_k+\gamma)}{\Gamma(\gamma)}$$
            (\cite{fruhwirth2021generalized}, Algorithm 3) with proposal distribution $\text{Normal}(\gamma,\sigma^2_\gamma)$.
        \item \textit{Update $\pi$}: Sample $\pi|\gamma,Z\propto\text{Dirichlet}(\gamma+N_1,\dots,\gamma+N_K)$.
    \end{enumerate}
\end{enumerate}
\bigbreak

\subsection{BTL-Binomial MAP Estimation}

In this section, we present a Bayesian Expectation Maximization (EM) algorithm to obtain maximum \textit{a posteriori} (MAP) estimates of a BTL-Binomial latent class mixture model under fixed $K$. 

\subsubsection{Data Likelihood}

To aid intuition, we state the loglikelihood after augmentation with latent classes and the objective function we seek to maximize: Assume the model presented in Equation \ref{eq:FixedK}. Let $z_{ik}=I\{Z_i=k\}$, where $I(\cdot)$ is the indicator function. Then, the loglikelihood of the preference data augmented with class indicators $Z$ can be written,
\begin{align*}
    \mathcal{L}&(\Pi,X,Z|\pi,p,\theta)\\
    &= \prod_{i=1}^I\prod_{k=1}^K\Bigg(\pi_k \times \prod_{r=1}^R \frac{e^{-\theta_k p_{\pi_i(r)k}}}{\sum_{j\in\mathcal{S}} e^{-\theta_k p_{jk}}-\sum_{s=1}^{r-1} e^{-\theta_k p_{\pi_i(s)k}}} \times \prod_{j=1}^J {M\choose x_{ij}}p_{jk}^{x_{ij}}(1-p_{jk})^{M-x_{ij}}\Bigg)^{z_{ik}}.
\end{align*}

We seek MAP estimates $(\hat\gamma,\hat\pi,\hat p,\hat\theta)$ defined as,
\begin{align}
    \hat\gamma,\hat\pi,\hat p,\hat\theta &= \underset{\gamma,\pi,p,\theta}{\arg\max}\Big( \log\mathcal{L}(\Pi,X,Z|\pi,p,\theta) + \log f(\gamma,\pi,p,\theta)\Big)\label{eq:em_objective},
\end{align}
where $f(\gamma,\pi,p,\theta)$ specifies the joint density of the prior distributions on $\gamma$, $\pi$, $p$, and $\theta$. 

\subsubsection{EM Algorithm}

Given the unknown latent class indicators $Z$, we use a Bayesian EM algorithm that is defined by the standard Expectation (E) and Maximization (M) steps, iterated until convergence of the objective function in Equation \ref{eq:em_objective}. The algorithm is presented below:

\noindent \underline{\textbf{Algorithm 3: MAP estimation of BTL-Binomial Latent Class Mixture Model}}\\
\noindent\textbf{Data:} $\Pi,X,K$\\
\noindent\textbf{Hyperparameters:} $\xi_1,\xi_2,a,b,\gamma_1,\gamma_2$\\
\noindent\textbf{Parameters:} $tol>0$
\begin{enumerate}
    \item Initialize parameters $\gamma$, $\pi_k$, $p_k$, and $\theta_k$ for $k\in\{1,\dots,K\}$. 
    \item Until the change in objective function (Equation \ref{eq:em_objective}) convergences to below $tol$,
    \begin{enumerate}
        \item \textit{E-Step}: For each $i=1,\dots,I$ and $k=1,\dots,K$, calculate the expected class membership indicators,
        \begin{align*}
            \hat z_{ik} &= \mathbb{E}_{\pi,p,\theta}[z_{ik}|\Pi_i,X_i]\\
            &= \frac{ \pi_kP[X_i,\Pi_i|Z_i=k,p_k,\theta_k]}{\sum_{k'=1}^K \pi_{k'}P[X_i,\Pi_i|Z_i=k',p_{k'},\theta_{k'}]}
        \end{align*}
        \item \textit{M-Step}: Maximize the unknown parameters sequentially:
        \begin{enumerate}
            \item Update $\pi$ subject to the constraint that $\sum_{k=1}^K\pi_k^{(t)}=1$:
           \begin{align*}
                \pi_k &= \frac{\gamma-1+\sum_{i=1}^I \hat z_{ik}}{K\gamma-K+I}.
            \end{align*}
            \item Update $\gamma$ according to,
            \begin{align*}
                \gamma &= \underset{\gamma}{\arg\max} \Big(\log f(\pi|\gamma) + \log f(\gamma)\Big)\\
                &= \underset{\gamma}{\arg\max}\Bigg(\log\Gamma(\gamma K)-K\log\Gamma(\gamma)+(\xi_1-1)\log\gamma + \gamma\Big(-\xi_2+\sum_{k=1}^K\log\pi_k\Big)\Bigg)
            \end{align*}
            via univariate numerical optimization. We use the function \texttt{optimize} in base \texttt{R} \citep{BaseR}.
            \item For each $k=1,\dots,K$, update $(p_k^{(t)},\theta_k^{(t)})$ according to,
            \begin{align*}
                (p_k^{(t)},\theta_k^{(t)})&=\underset{p_k,\theta_k}{\arg\max}\Bigg(\log f(p_k,\theta_k) + \sum_{i=1}^I \hat z_{ik}\log \Big(\text{BTL-Binomial}(\Pi_i,X_i|p_k,\theta_k)\Big)\Bigg)
            \end{align*}
            via multivariate numerical optimization. We use the function \texttt{optim} in base \texttt{R} with the method L-BFGS \citep{BaseR,byrd1995limited}.
        \end{enumerate}
    \end{enumerate}
\end{enumerate}

\subsubsection{Obtaining Frequentist Estimators}

To obtain frequentist maximum likelihood estimators, hyperparameters are available for each prior distribution. Setting $\xi_1=1$ and $\xi_2=0$ yields a flat but improper prior on $\gamma$. Alternatively, one may set $\gamma=1$ and remove step 2(b)(ii). Setting $a=b=1$ yields a flat and proper prior on each $p_{jk}$, and setting $\gamma_1=1$ and $\gamma_2=0$ yields a flat but improper prior on each $\theta_k$.

\section{Additional Application Results}\label{appendix:results}

\subsection{Paper Selection in Large Academic Conferences}

\subsubsection{Bias and Consistency of MAP Estimates}

We present the bias and consistency of MAP estimates in order to demonstrate good statistical properties of the BTL-Binomial model. We examine bias in Figure \ref{fig:simulation_accuracy} and consistency in Figure \ref{fig:simulation_error}.

\begin{figure}[h!!]
    \centering
    \includegraphics[width=\textwidth]{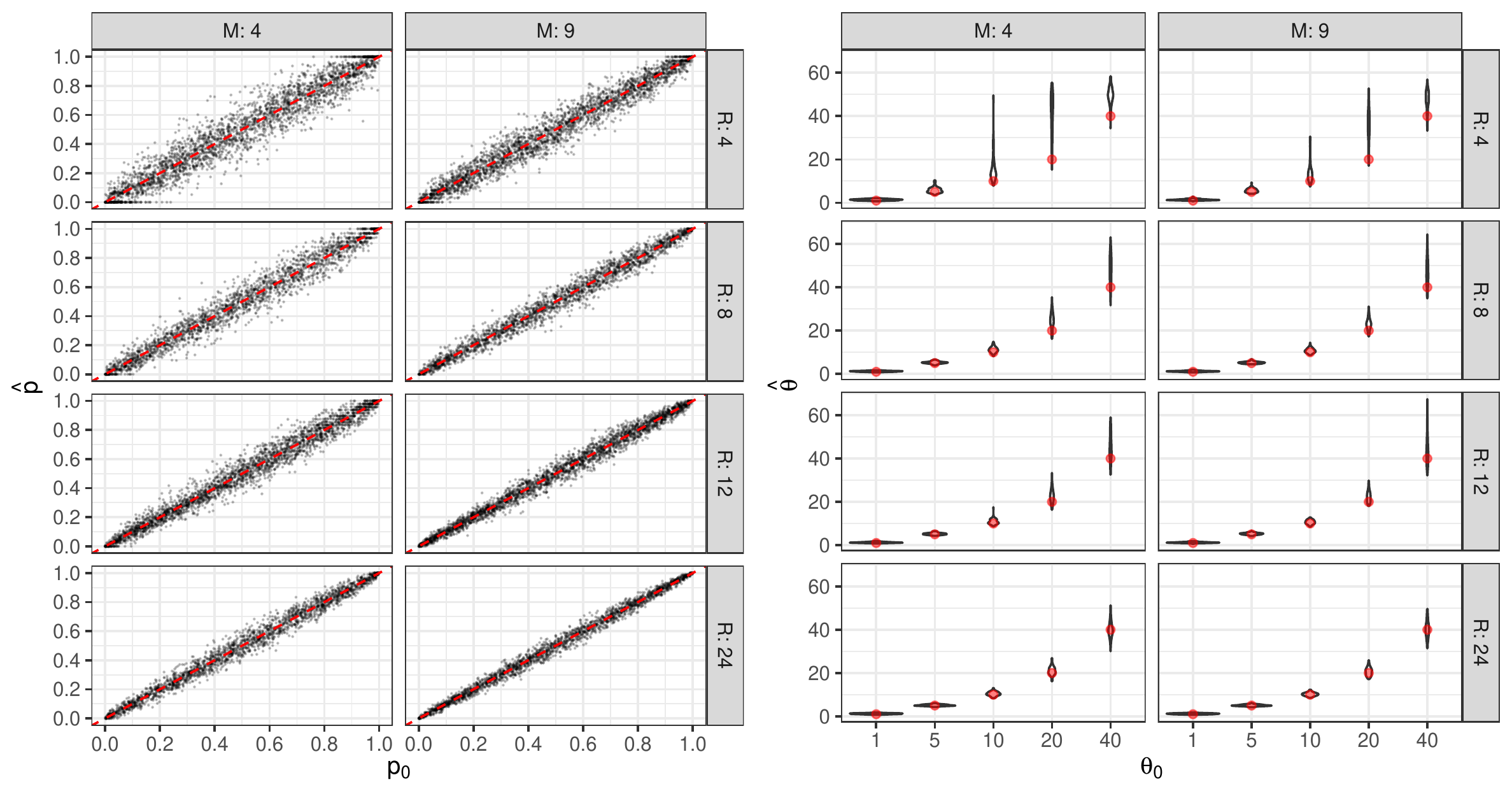}
    \caption{Scatterplot of MAP estimates $\hat p$ against true $p_0$ (left) and violin plots of MAP estimates $\hat \theta$ against true $\theta_0$ (right) under various $M$ and $R$.}
    \label{fig:simulation_accuracy}
\end{figure}
\begin{figure}[h!!]
    \centering
    \includegraphics[width=\textwidth]{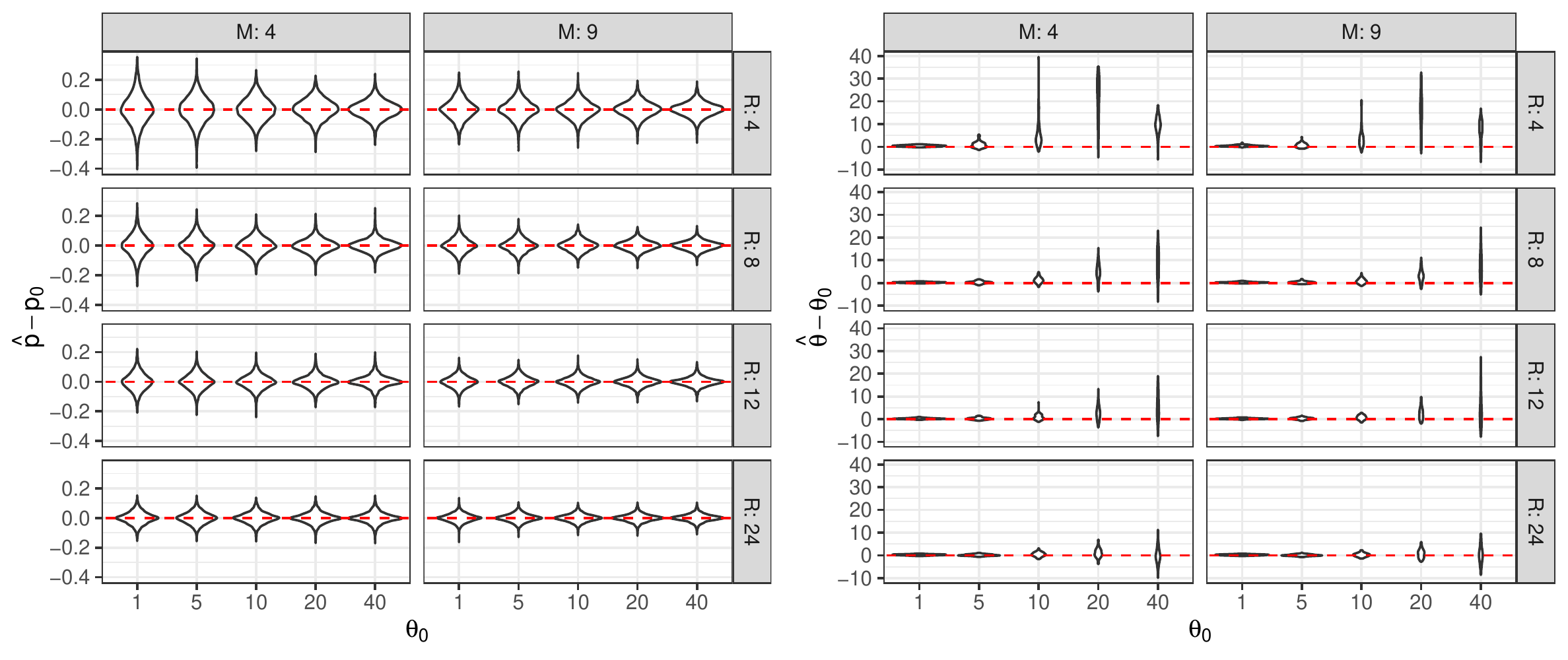}
    \caption{Violin plots of estimation error for $\hat p$, calculated $\hat p-p_0$ (left) and estimation error for $\hat\theta$, calculated $\hat\theta-\theta_0$ (right) under various $M$ and $R$. See the main body for simulation details.}
    \label{fig:simulation_error}
\end{figure}

We observe in the left panels of Figures \ref{fig:simulation_accuracy} and \ref{fig:simulation_error} that estimation of $p$ appears to be unbiased and consistent in both $M$ and $R$, regardless of $\theta_0$. Estimation of $p$ is central for understanding the quality of each paper and consensus ranking of the reviewers. Thus, the apparent unbiasedness and consistency of $\hat p$ suggest accurate estimation of group preferences regardless of $M$, $R$, and $\theta_0$. For $\theta$, the right panels demonstrate potentially biased estimation. This is unsurprising given a similar result for the corresponding parameter in the related Mallows-Binomial model \citep{pearce2022unified}. Estimation accuracy appears worst when $\theta_0$ is very large, which often leads to perfect uniformity of observed rankings. In such cases, estimation of $\theta$ is most difficult; estimation is more accurate when $\theta$ is small. We notice that $\hat\theta$ appears consistent in $R$ but not $M$, which makes sense given that $M$ only relates to the rating scale while $\theta$ is only applicable to ranking consistency. Although the potentially biased estimation of $\theta$ is disappointing, it may not have much practical impact in the present setting since it corresponds to the strength of consensus and not the relative or ordered preferences of the group, which are paramount for deciding which papers to accept to the conference.

\subsection{Proposal Selection in Grant Panel Review under Heterogeneity}

\subsubsection{Hyperparameter and Algorithm Parameter Settings}

Given the small sample size, we assign prior weight primarily to $K^+\in\{1,2,3\}$. Thus, we choose $\lambda = 1$, $\xi_1=2$, and $\xi_2=3$. The effect of these choices on the prior distribution of $K^+$ can be seen in Figure~\ref{priorKplus_AIBS}. We set $a=2.50$ and $b=3.77$ using an empirical Bayes approach, in which we fit a Beta distribution to the observed ratings after normalization to the unit interval based on maximum moment estimators of the first two moments. We set $\gamma_1=10$ and $\gamma_2=0.5$ to provide substantial weight to values of $\theta\in[5,35]$. We carry out Algorithm 1 with $B^\text{Gibbs}=1000$ and $B^\text{MH}=10$, $K_\text{start}=I=17$, $\sigma_p^2=0.05$, $\sigma_\theta^2=3$, and $\sigma_\gamma^2=.5$. The first half of the total 10,000 iterations were removed as burn-in.

\begin{figure}[h!!]
    \centering
    \includegraphics[width=\textwidth]{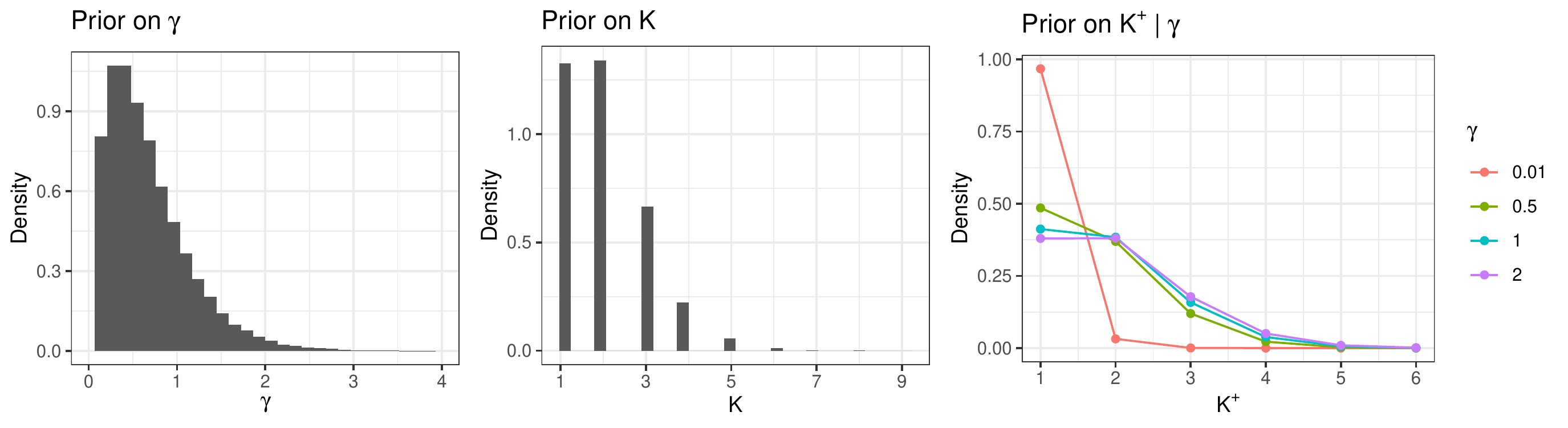}
    \caption{Prior distribution on $\gamma$, $K$, and $K^+$ given $\lambda=1$ and $\gamma\in\{0.01,0.5,1,2\}$.}
    \label{priorKplus_AIBS}
\end{figure}

\subsubsection{Further Estimation Results}

Below, we present an augmented version of Figure 3 in the main body of the paper that additionally includes a posterior summary of $\gamma$.
\begin{figure}[ht]
    \centering
    \includegraphics[width=\textwidth]{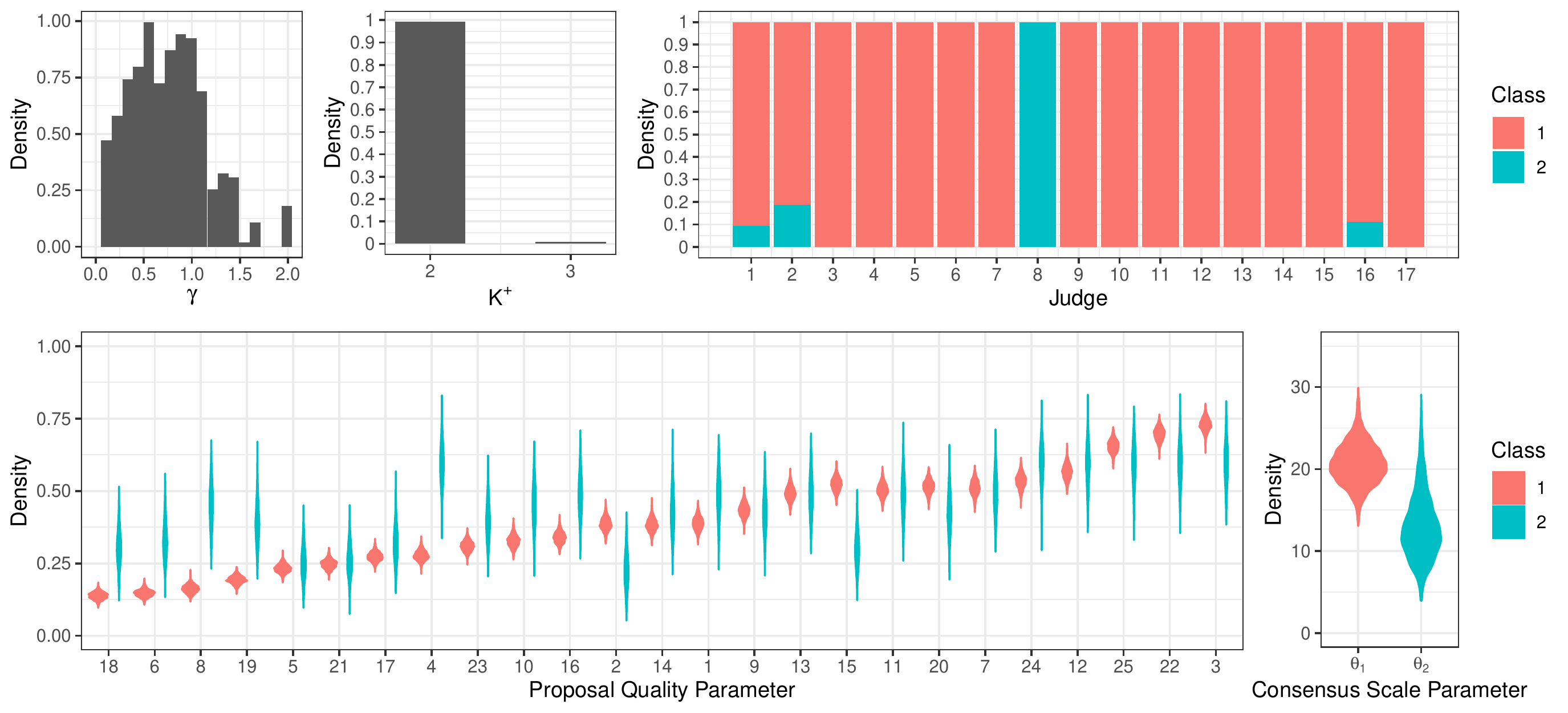}
    \caption{Posterior summaries of $\gamma$ (\textit{top-left}), $K^+$ (\textit{top-middle}), class membership probabilities given $K^+=2$ (\textit{top-right}); and preference parameters given $K^+=2$ (\textit{bottom}).}
\end{figure}

\newpage
\subsubsection{Goodness-of-Fit and Trace Plots}

Below we display goodness-of-fit and trace plots for Setting 2. We find the results to be satisfactory.

\begin{figure}[h!!]
    \centering
    \includegraphics[width=\textwidth]{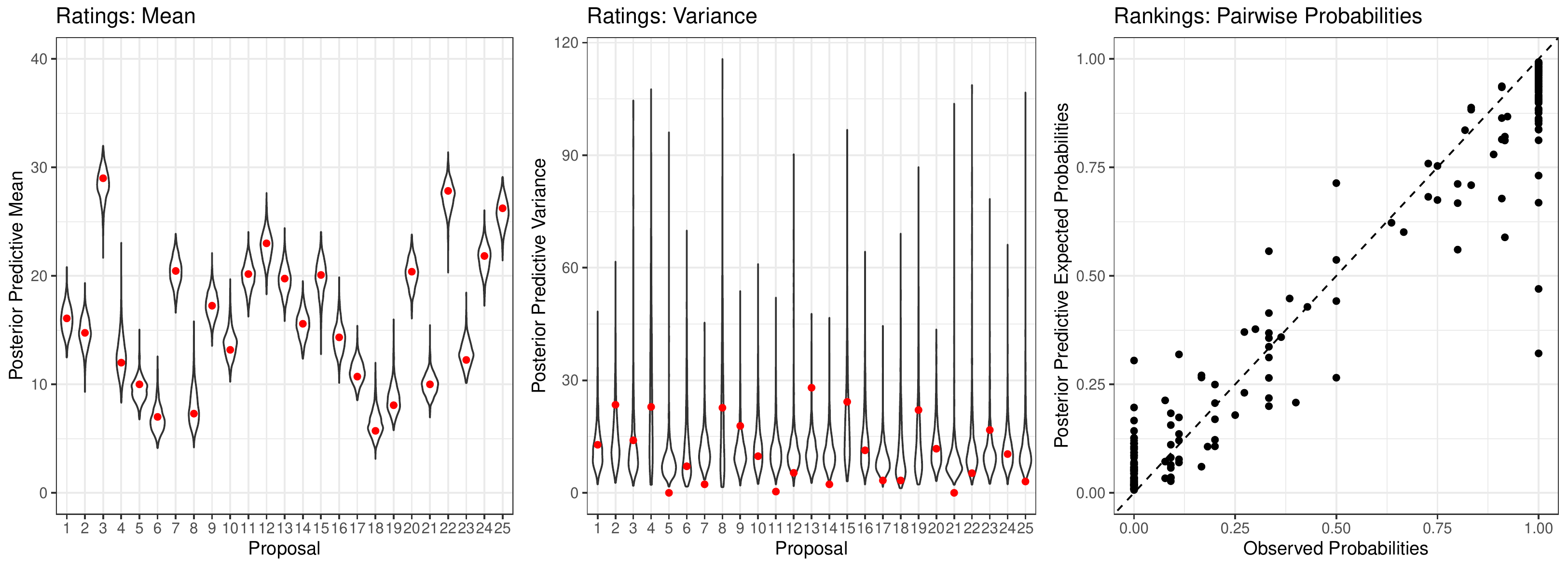}
    \caption{AIBS Goodness-of-Fit Results. Red dots represent the observed posterior mean (left) or variance (center).}
\end{figure}

\begin{figure}[h!!]
    \centering
    \includegraphics[width=\textwidth]{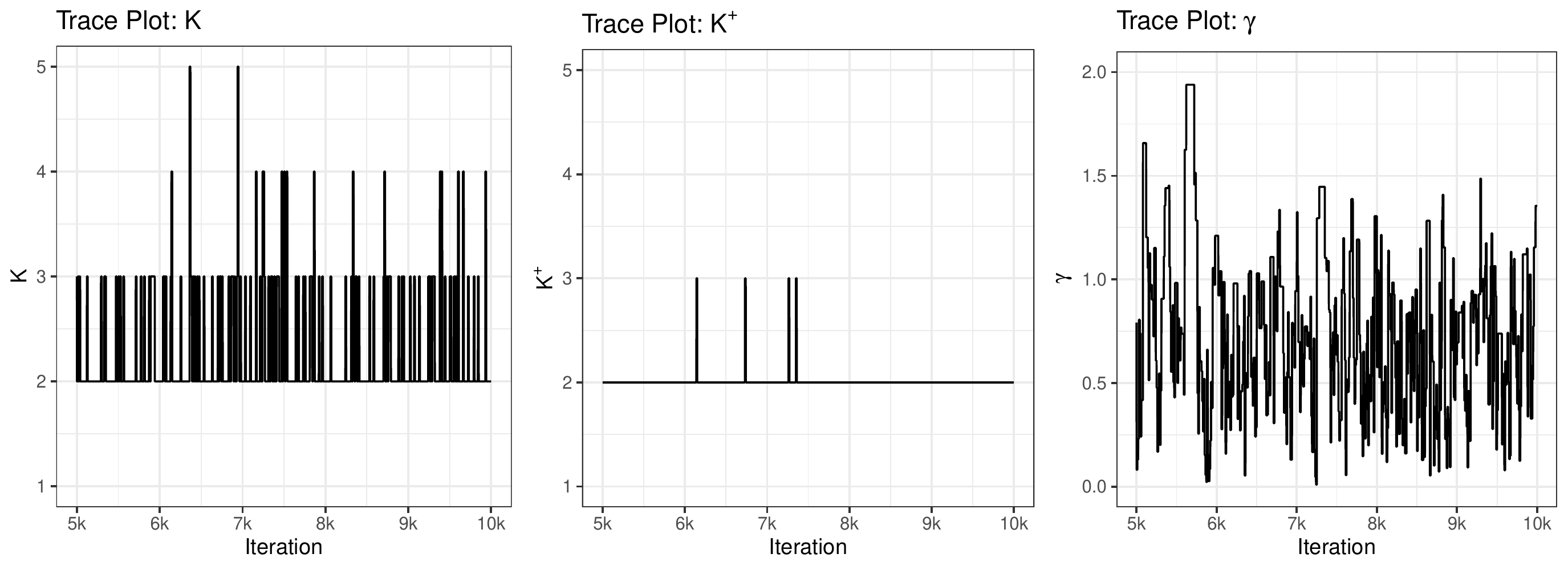}
    \includegraphics[width=\textwidth]{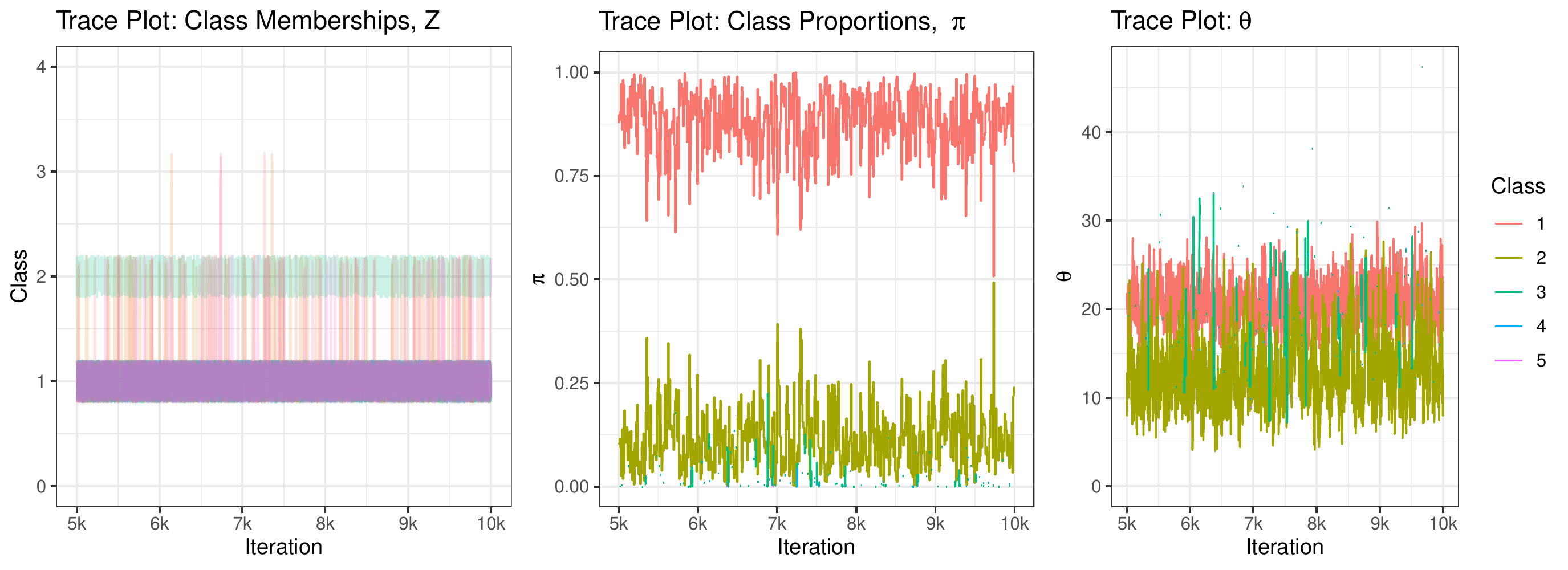}
    \caption{AIBS Trace Plots of $K$, $K^+$, $\gamma$, $Z$, $\pi$, and $\theta$}
\end{figure}
\begin{figure}[h!!]
    \centering
    \includegraphics[width=\textwidth]{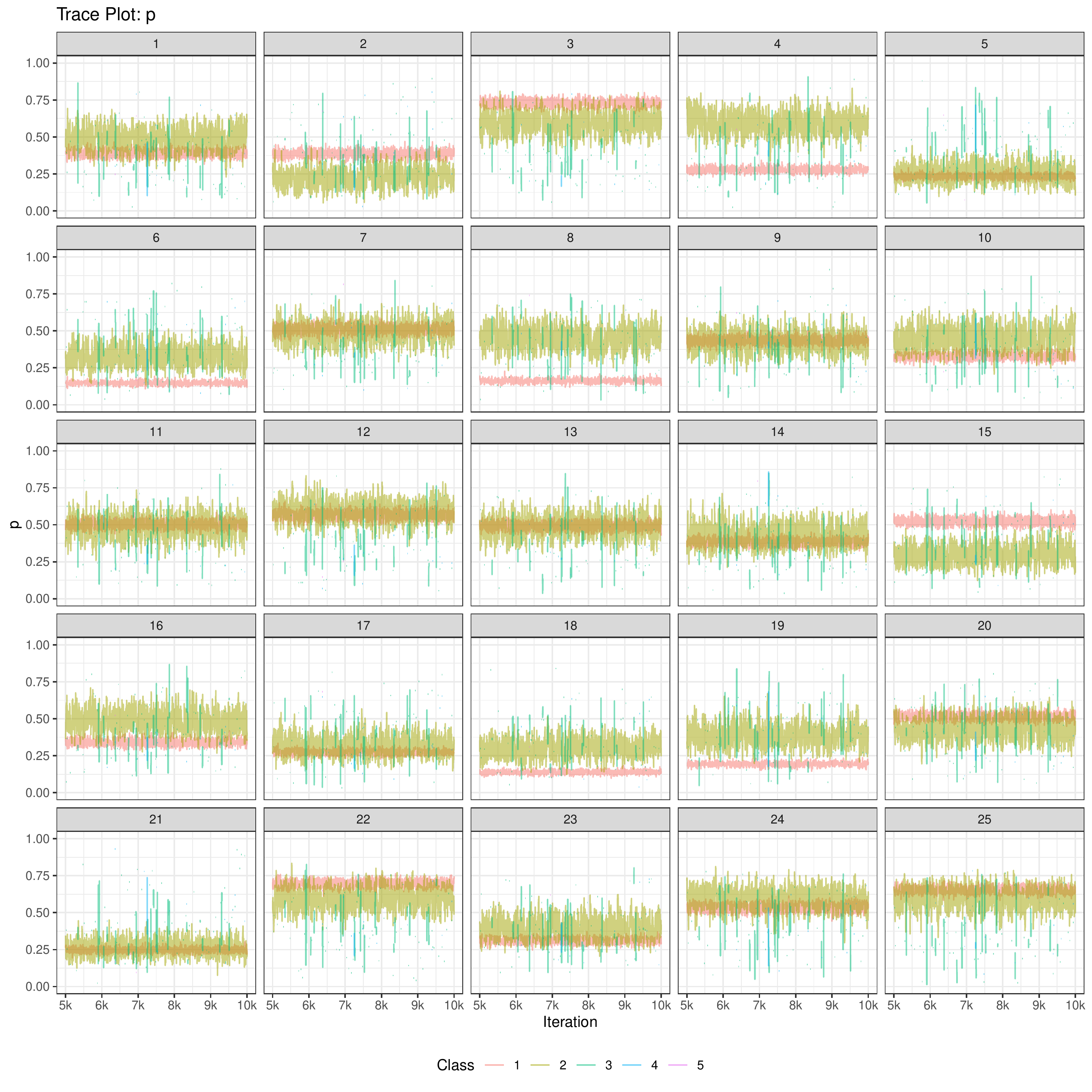}
    \caption{AIBS Trace Plots of $p$}
\end{figure}

\clearpage
\subsection{Modeling Complex Survey Data under Heterogeneity}

\subsubsection{Hyperparameter and Algorithm Parameter Settings}

To aid interpretability given the large sample size, we assign prior weight primarily to $K^+\in\{5,\dots,10\}$ using $\lambda = 7$, $\xi_1=2$, and $\xi_2=3$. The effect of these choices on the prior distribution of $\gamma$, $K$, and $K+$ can be seen in Figure \ref{priorKplus_Sushi}. We set $a=0.26$ and $b=0.77$ using an empirical Bayes approach, in which we fit a Beta distribution to the observed ratings after normalization to the unit interval based on maximum moment estimators of the first two moments. We set $\gamma_1=20$ and $\gamma_2=1$ to provide substantial weight to values of $\theta\in[10,30]$. We carry out Algorithm 1 with $B^\text{Gibbs}=2500$ and $B^\text{MH}=10$, $K_\text{start}=1$, $\sigma_p^2=0.1$, $\sigma_\theta^2=3$, and $\sigma_\gamma^2=0.3$. The first half of the total 25,000 iterations were removed as burn-in. 
\begin{figure}[h!!]
    \centering
    \includegraphics[width=\textwidth]{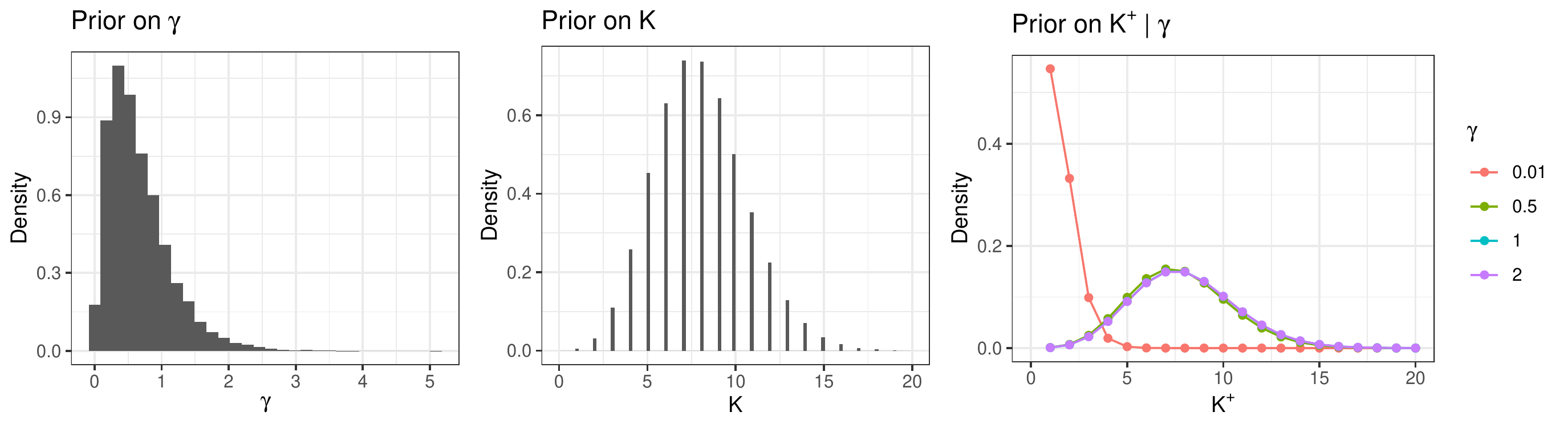}
    \caption{Prior distribution on $\gamma$, $K$, and $K^+$ given $\lambda=7$, conditional on $\gamma\in\{0.01,0.5,1,2\}$.}
    \label{priorKplus_Sushi}
\end{figure}

\subsubsection{Further Estimation Results}

First, we present posterior summaries of $K^+$, $\gamma$, and $\pi$, where classes are ordered by size. We see that an 8-class model has very high posterior probability, which leads us to present results conditional on $K^+ = 8$. 

\begin{figure}[h!]
    \centering
    \includegraphics[width=\textwidth]{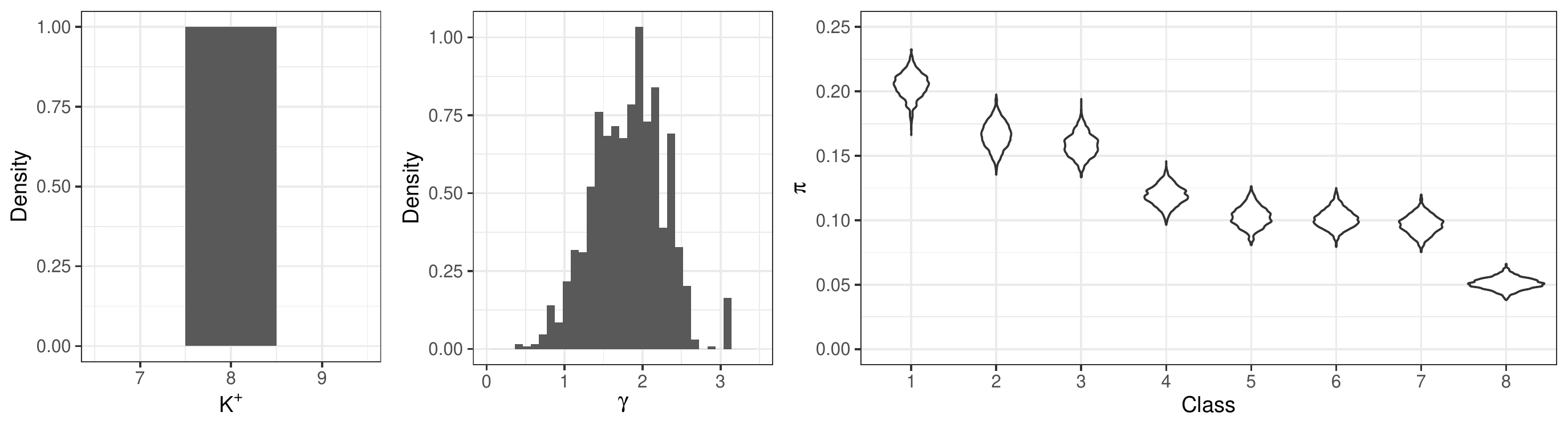}
    \caption{Posterior distributions of $K+$ (\textit{left}), $\gamma$ (\textit{center}), and $\pi$ (\textit{right}).}
    \label{fig:sushi_res1}
\end{figure}

Second, we present posterior probabilities of shared class membership across survey respondents (Figure \ref{fig:sushi_res2}). We do not display class membership probabilities by respondent due to the very large number of respondents. On the x- and y-axes are survey respondents, with order determined to keep similarly-clustered respondents together, as determined by the \texttt{salso} package \citep{salso}. The color indicates cluster similarity via the posterior probability of shared class membership (white represents low probability; black represents high probability). We notice high within-class homogeneity with respect to clustering probability and relatively strong heterogeneity between classes. We take this as evidence that the algorithm is successful at distinguishing heterogeneous groups.

\begin{figure}[ht]
    \centering
    \includegraphics[width=.9\textwidth]{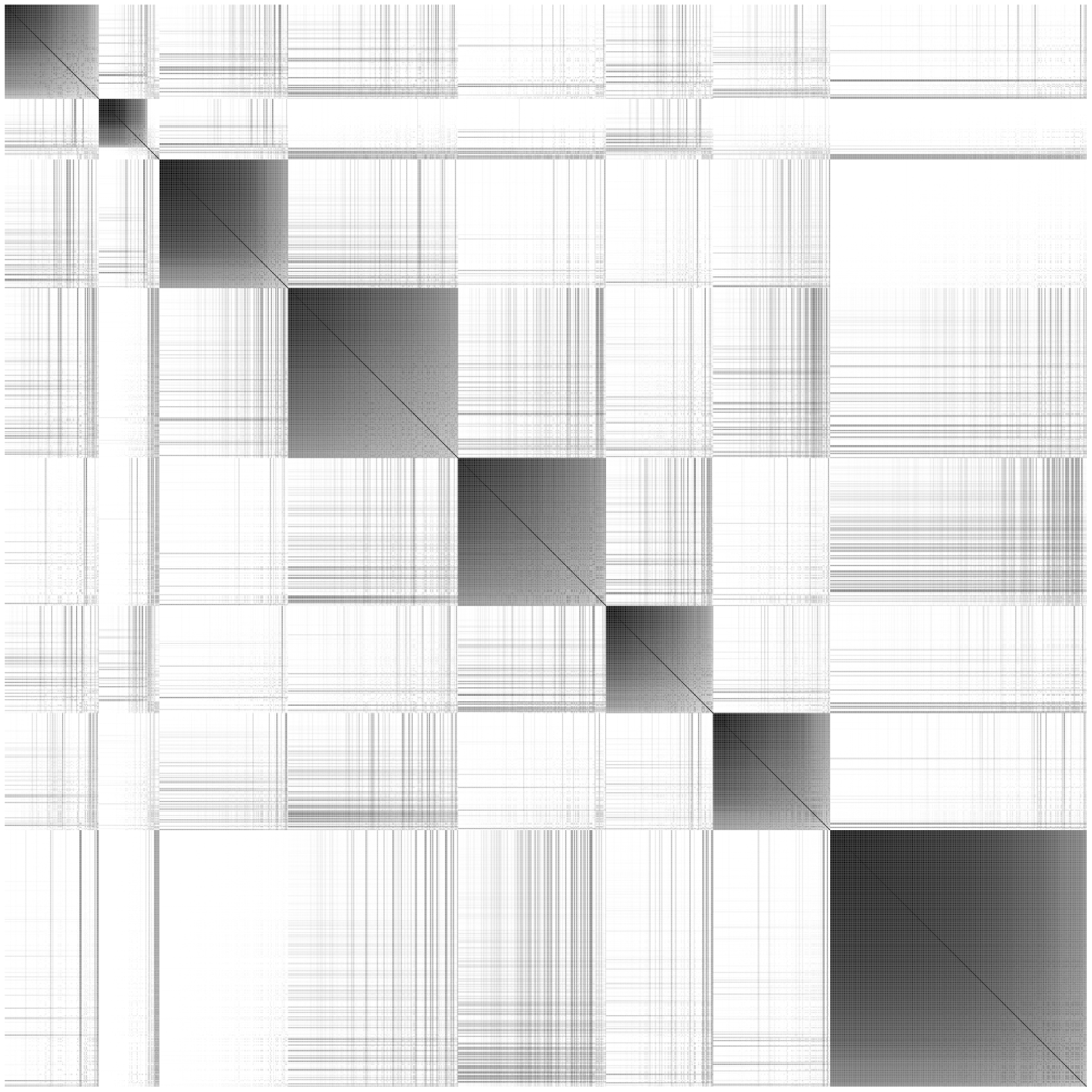}
    \caption{Similarity matrix of survey respondents by shared class membership. Survey respondents are on the x- and y-axes, ordered according to display clusters cohesively. White (black) represents low (high) posterior probability that two respondents are in the same class.}
    \label{fig:sushi_res2}
\end{figure}

\subsubsection{Goodness-of-Fit and Trace Plots}

Below we display goodness-of-fit and trace plots for Setting 3. We notice that posterior means for ratings and pairwise probabilities for rankings appear satisfactory. The posterior predictive ratings variance appears to be low in comparison to the observed ratings variance, which is likely a result of providing strong prior probability to a relatively small number of clusters. Given the focus on interpretability, we find the results to be satisfactory.

\begin{figure}[h!!]
    \centering
    \includegraphics[width=\textwidth]{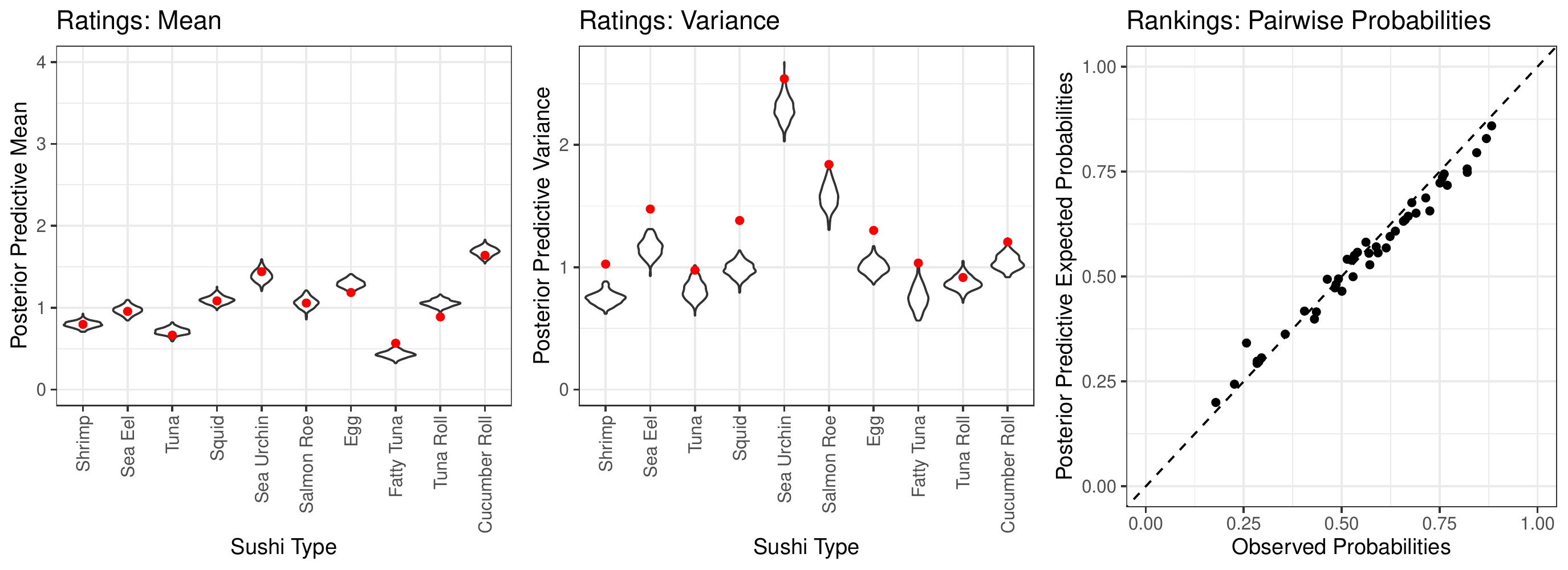}
    \caption{Sushi Data Goodness-of-Fit}
\end{figure}

\begin{figure}[h!!]
    \centering
    \includegraphics[width=\textwidth]{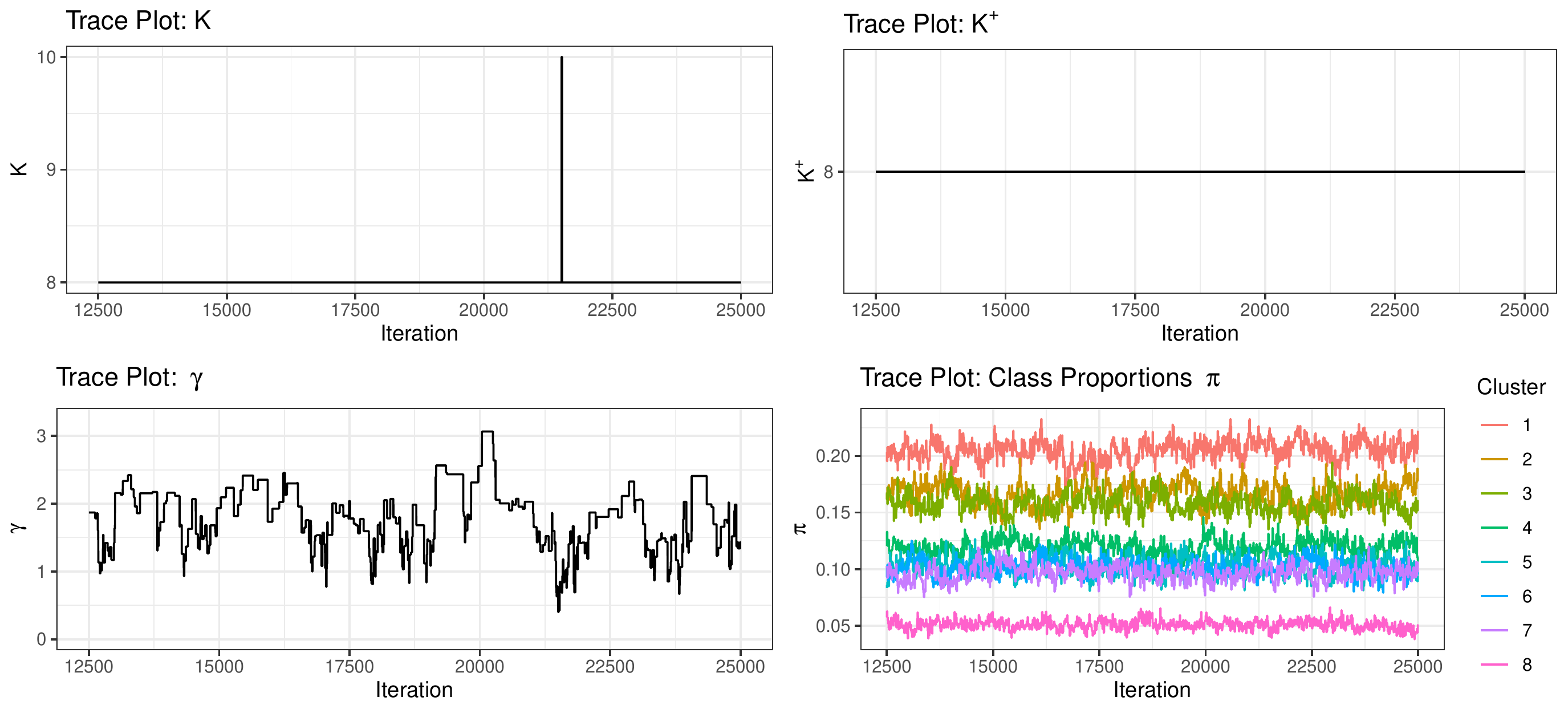}
    \caption{Sushi Trace Plots of $K$, $K^+$, $\gamma$, and $\pi$}
\end{figure}

\newpage
\section{Proof of Theorem 1}\label{appendix:proofs}

\begin{theorem}
Let $M$, $J$, and $R$ be fixed and positive integers such that $R\leq J$. Then the BTL-Binomial($p,\theta$) model is identifiable.
\end{theorem}
\begin{proof}
Let $P_{p,\theta}$ denote the probability distribution of scores $x$ and rankings $\pi$ under a BTL-Binomial($p,\theta$) model. Let $\theta_1,\theta_2>0$ and $p_1,p_2\in[0,1]^J$ such that $P_{\theta_1,p_1}=P_{\theta_2,p_2}$. Given $M$, the standard Binomial distribution is identifiable. Thus, for $j=1,\dots,J$ and arbitrary $x_j$, we know that $P_{\theta_1,p_1}=P_{\theta_2,p_2}$ if and only if $p_1=p_2$. Continuing under this assumption, it remains to show that $P_{p_1,\theta_1}=P_{p_1\theta_2} \iff \theta_1=\theta_2$. Note that,
\begin{align*}
    P_{p_1,\theta_1}&=P_{p_1,\theta_2}\\
    \iff&\prod_{r=1}^R \frac{e^{-\theta_1 p_{1\pi(r)}}}{\sum_{j\in\mathcal{S}} e^{-\theta_1 p_{1j}}-\sum_{s=1}^{r-1} e^{-\theta_1 p_{1\pi(s)}}}=\prod_{r=1}^R \frac{e^{-\theta_2 p_{1\pi(r)}}}{\sum_{j\in\mathcal{S}} e^{-\theta_2 p_{1j}}-\sum_{s=1}^{r-1} e^{-\theta_2 p_{1\pi(s)}}}\\
    \iff& 0 = \sum_{r=1}^R \Big[p_{1\pi(r)}(\theta_2-\theta_1)+ \log\Big(\frac{\sum_{j\in\mathcal{S}} e^{-\theta_2 p_{1j}}-\sum_{s=1}^{r-1} e^{-\theta_2 p_{1\pi(s)}}}{\sum_{j\in\mathcal{S}} e^{-\theta_1 p_{1j}}-\sum_{s=1}^{r-1} e^{-\theta_1 p_{1\pi(s)}}}\Big)\Big]
\end{align*}
which for arbitrary $\pi$ will be true only when $\theta_1=\theta_2$, as desired.
\end{proof}

\section{Code and Data Access}
\label{appendix:implementation}
All algorithms described herein have been implemented in \texttt{R}. The code to implement our method and reproduce our analyses is available upon request. The data used in this paper is either publicly available or available upon request:
\begin{itemize}
    \item Setting 1: The simulated data is available upon request.
    \item Setting 2: The peer review data used in this study is publicly available in \cite{Gallo2023} or at \url{https://doi.org/10.6084/m9.figshare.19692223.v1}.
    \item Setting 3: The sushi preference data used in this study is publicly available in \cite{kamishima2003nantonac} or at \url{https://www.kamishima.net/sushi/}.
\end{itemize}

\end{document}